\documentclass[letterpaper, 10 pt, conference]{ieeeconf} 

\IEEEoverridecommandlockouts                             
\usepackage{amsmath} 
\usepackage{amssymb}  

\newtheorem{theorem}{Theorem}[section]

\newtheorem{proposition}[theorem]{Proposition}
\newtheorem{remark}{Remark}[section]

\usepackage{hyperref}

\usepackage{graphicx}
\usepackage{subfigure}
\usepackage{textcomp}
\usepackage{float}

\usepackage[linesnumbered,boxed,ruled,commentsnumbered]{algorithm2e}
\SetKwRepeat{Do}{do}{while}%
\usepackage{color}

\title{\LARGE \bf
Rollout-Based Charging Strategy for Electric Trucks\\ 
\vspace{1.5pt}

with Hours-of-Service Regulations (Extended Version)}

\author{Ting Bai,~Yuchao Li,~Karl H. Johansson,~\IEEEmembership{Fellow, IEEE}, and Jonas M{\aa}rtensson
\thanks{This work is supported by the Swedish Research Council Distinguished Professor (Grant Number: 2017-01078), the Knut and Alice Wallenberg Foundation, and the Swedish Strategic Research Foundation CLAS (Grant Number: RIT17-0046).}
\thanks{The authors are with the Integrated Transport Research Lab and Division of Decision and Control Systems, KTH Royal Institute of Technology, SE-100 44 Stockholm, Sweden. They are also affiliated with Digital Futures. E-mails: \{{\tt\small tingbai, yuchao, kallej, jonas1\}@kth.se}}}

\begin{document}

\maketitle
\thispagestyle{empty}
\pagestyle{empty}

\begin{abstract}
Freight drivers of electric trucks need to design charging strategies for where and how long to recharge the truck in order to complete delivery missions on time. Moreover, the charging strategies should be aligned with drivers' driving and rest time regulations, known as hours-of-service (HoS) regulations. This letter studies the optimal charging problems of electric trucks with delivery deadlines under HoS constraints. We assume that a collection of charging and rest stations is given along a pre-planned route with known detours and that the problem data are deterministic. The goal is to minimize the total cost associated with the charging and rest decisions during the entire trip. This problem is formulated as a mixed integer program with bilinear constraints, resulting in a high computational load when applying exact solution approaches. To obtain real-time solutions, we develop a rollout-based approximate scheme, which scales linearly with the number of stations while offering solid performance guarantees. We perform simulation studies over the Swedish road network based on realistic truck data. The results show that our rollout-based approach provides near-optimal solutions to the problem in various conditions while cutting the computational time drastically. 
\end{abstract}

\section{Introduction}\label{Section I}
Vehicle electrification is becoming mainstream globally to reduce carbon emissions and achieve sustainable transportation~\cite{GlobalEVSales}. In particular, road freight electrification is crucial for reducing greenhouse gas emissions caused by diesel-powered trucks in freight operations, which are responsible for around $25\%$ of vehicle-related carbon emissions in Europe~\cite{siskos2019assessing}. However, the process of freight truck electrification today is lagging far behind that of electric passenger vehicles~\cite{gloableev}. A major concern with electrifying trucks, among others, is their limited driving ranges, known as \emph{range anxiety}. Currently, the average travel range of a commercial electric truck on a full battery varies between $200$ and $600$ kilometers, depending on diverse truckloads and battery capacities~\cite{wassiliadis2021review}. This is typically insufficient to sustain trucks to complete their delivery missions without stopping and refilling batteries, especially for long-haul journeys. To diminish range anxiety, increase electric truck adoption, and accelerate road freight electrification, reliable and efficient charging strategies are needed. In addition to charging batteries, truck drivers also need to stop and take rests during trips to avoid driving fatigue. The so-called hours-of-service (HoS) regulations~\cite{poliak2018social} address exactly this issue and put restrictions on how long one can drive consecutively without rest, as well as during one day. As a result, charging strategies for electric trucks should be designed not only for mission completion but also to align with the HoS regulations.

To date, there have been extensive works developing viable charging strategies. A majority of these approaches integrate charging stops into conventional routing problems and minimize the travel time or energy required on the route, as in \cite{storandt2012quick,schneider2014electric,huber2020optimization}. The authors in \cite{cussigh2019optimal} propose an optimal driving and charging strategy for electric vehicles. It includes the driving speed as an additional control variable when minimizing the total travel time. However, none of these works incorporates the HoS regulations in optimal charging problems. To the best of our knowledge, \cite{zahringer2022time} is the first to incorporate today's HoS regulations in their charging strategy, which is obtained via a genetic algorithm. Our work differs from the approach in \cite{zahringer2022time} in two major ways. Firstly, we model the route and optimal charging problem in a more general framework, allowing for multiple rests within the maximum daily driving time before delivery deadlines. Secondly, we develop an online solution scheme that allows real-time optimization to deal with travel time uncertainties or model mismatches, as opposed to the genetic algorithm, which is offline and time-consuming.

To enable a real-time solution, we rely on the idea of \emph{rollout}, which refers to the process of simulating a known solution. Proposed first in \cite{tesauro1996line} for addressing backgammon, rollout has since been extended for combinatorial optimization \cite{bertsekas1997rollout} and trajectory-constrained problems \cite{bertsekas2005rollout}, to name a few. Our scheme is modified from the methods introduced in \cite{bertsekas2005rollout} and \cite[Section~3.4]{bertsekas2020rollout}, where an improved solution is computed based on a known solution. As shown in \cite{bertsekas2020rollout}, the rollout scheme can be viewed as one step of Newton's method applied to solve the optimization problem with the initial guess supplied by the known solution. In view of the fast convergence rate of Newton's method, the solution provided by rollout has substantial improvement from the known solution, which is consistent with many empirical studies \cite[p.~136]{bertsekas2020rollout}.

In this letter, we study the optimal charging strategy for electric trucks with realistic HoS regulations. In particular, we consider an electrified transportation system where every electric truck has a pre-planned route for a delivery task. With the knowledge of a collection of charging and rest stations available along the route, the truck driver could design an optimal charging strategy to determine where and how long to recharge the truck and take rests so that the extra operational costs due to the charging and rest decisions are minimized. The main contributions of this letter are: \emph{i)} we model the optimal charging problems of electric trucks as mixed integer programs with \emph{bilinear constraints} that incorporate the delivery deadlines and HoS regulations; \emph{ii)} a rollout-based approximate solution method, as well as its variants, is developed for addressing the problem, through which the computational demands required by exact solution approaches are significantly decreased while offering solid performance guarantees. Simulation studies performed over the Swedish road network using realistic truck data illustrate the effectiveness of the developed method.

\section{Problem Formulation}\label{Section II}
\subsection{Route Model}
We consider simplified route models of electric trucks. As illustrated in Fig.~\ref{Fig.1}, given a pre-planned route of a truck between its origin $O$ and destination $D$, we assume that $N$ \emph{charging and rest stations} are available along the route, denoted as $S_k$, $k\!=\!0,\dots, N\!-\!1$. The ramp along the pre-planned route leading to station $S_k$ with the shortest detour is denoted as $r_k$. For simplicity, the destination is referred to as ramp $r_N$. The travel time required for the truck to take a detour between its ramp $r_k$ and the station $S_k$ (one-way) is denoted as $d_k$. Moreover, in the pre-planned route, the travel time on the route segment connecting $r_k$ and $r_{k+1}$ is denoted by $\tau_{k+1}$, $k\!=\!0,\dots, N\!-\!2$. Particularly, the truck's travel times from its origin to its first ramp $r_0$, and from its ramp $r_{N-1}$ to the destination are denoted by $\tau_0$ and $\tau_N$, respectively. 

\begin{figure}[t!]
\centering
     \includegraphics[width=0.95\linewidth]{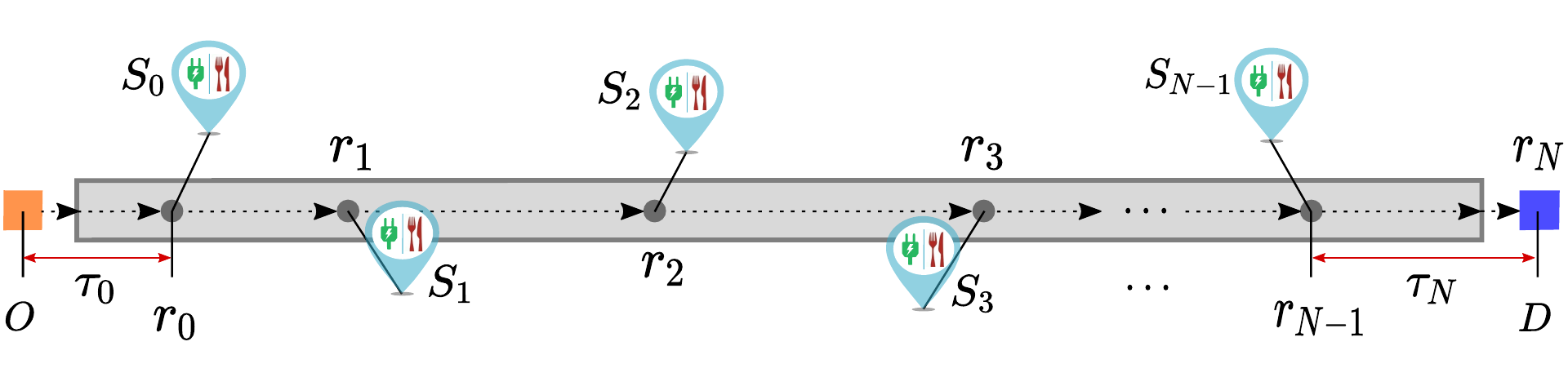}
      \caption{A simplified route model of electric trucks, where each charging and rest station, denoted by $S_k$, provides both charging and rest services. Each ramp in the route leading to $S_k$ with the shortest detour is denoted by $r_k$ and shown by a grey node, where $k\!=\!0,\dots, N\!-\!1$.}
      \label{Fig.1}
\end{figure}

Note that in our route model, each station provides both charging and rest services, allowing drivers to take a rest while the truck is charging. This is aligned with the charging infrastructure building plans nowadays~\cite{speth2022public} as it saves time in logistics. Nevertheless, our problem formulation and solution approach to be introduced later also apply to other cases where some stations provide only charging or rest service.

\subsection{Battery Energy and Consecutive Driving Time}
Given the travel and detour times on route, the initial battery energy, and the delivery deadline, a truck driver could design a charging strategy to complete the delivery mission on time subject to the HoS regulations. The decisions involved in the charging strategy include: (i) whether to charge the truck at $S_k$, $k\!=\!0,\dots, N\!-\!1$; (ii) whether to rest at $S_k$, $k\!=\!0,\dots, N\!-\!1$; and (iii) how long to charge the truck at $S_k$ once decided to charge there. These decisions can be represented by the variables
\vspace{-1pt}
\begin{equation}
    \label{eq:control}
    b_k,\Tilde{b}_k\!\in\!\{0,1\},\;t_k\!\in\!\Re_+,\quad k=0,\dots,N\!-\!1,
\end{equation}
where $b_k\!=\!1$ means charging at $S_k$ and $0$ otherwise, $t_k$ is the planned charging time at $S_k$ if $b_k\!=\!1$, and $\Re_+$ contains nonnegative reals. Similarly, $\Tilde{b}_k\!=\!1$ means resting at $S_k$ and $0$ otherwise. These decisions affect how the battery energy and the consecutive travel time vary when arriving at different ramps, as we introduce next.

To describe the battery dynamics, we denote by $e_k\!\in\! \Re_+$ the \emph{remaining energy in the battery} of the truck when it \emph{first} arrives at $r_k$, $k\!=\!0,\dots,N$. In addition, let $e_{\text{ini}}$ be the initial energy in the battery at the origin, and $\bar{P}$ be the battery consumption of the truck on route per travel time unit. The remaining energy $e_k$ can then be characterized as
\begin{subequations}
\label{eq:battery_energy}
\begin{align}
e_{k+1}&=e_k+b_{k}\Delta e_k-\bar{P}\big(2\bar{b}_kd_k\!+\!\tau_{k+1}\big)\textcolor{blue}{,}\label{Eq.2a}\\
\bar{b}_k&=b_k\lor\tilde{b}_k,\label{Eq.2b}
\end{align}
\end{subequations}
for $k\!=\!0,\dots, N\!-\!1$, with $e_{0}\!=\!e_{\text{ini}}\!-\!\bar{P}\tau_{0}$. In \eqref{Eq.2b}, $\bar{b}_k$ reflects whether the truck visits $S_k$ for charging or rest, where $\lor$ is logical or operator. In \eqref{Eq.2a}, $\Delta e_k$ denotes the charged energy at $S_k$.  As a linear approximation to the charging process, as adopted in \cite{liu2021optimal}, $\Delta e_k$ is modeled as
\begin{equation}
\label{eq:charging}
\Delta{e}_k=t_k\min\big\{P_k,P_{\max}\big\},\quad k=0,\dots, N\!-\!1,
\end{equation}
where $P_k$ denotes the charging power provided by $S_k$, and $P_{\max}$ is the maximum charging power that can be accepted by the battery of the truck.
\begin{figure}[t]
     \centering
     \includegraphics[width=0.95\linewidth]{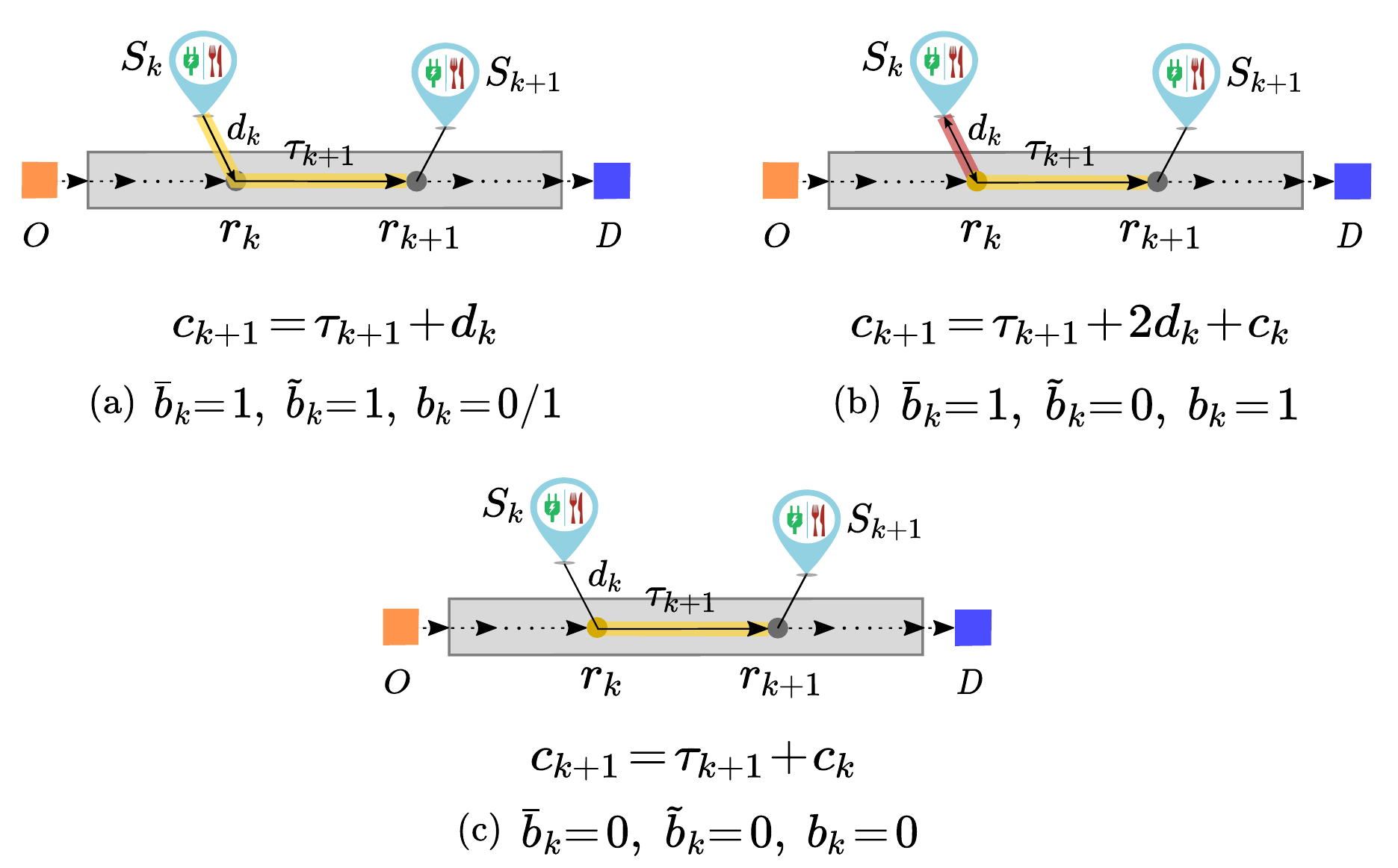}
      \caption{The consecutive driving time when first arriving at $r_{k+1}$.}
      \label{Fig.2}
\end{figure}

To ensure that the HoS regulations are followed under the designed strategy, we also keep track of consecutive driving time upon arriving at each ramp. For this purpose, let us denote by $c_k$, $k\!=\!0,\dots, N$, the consecutive driving time when arriving at $r_k$ for the first time. Its dynamics are given by
\begin{equation}
\label{eq:consecutive_t}
c_{k+1}=\tau_{k+1}+\bar{b}_kd_k+\big(1-\tilde{b}_k\big)\big(c_k+b_kd_k\big),
\end{equation}
for $k=0,\dots,N\!-\!1$, with $c_0\!=\!\tau_0$. Here\textcolor{blue}{,} $\tau_{k+1}\!+\!\bar{b}_kd_k$ is the driving time from $S_k$ to $r_{k+1}$ when $\bar{b}_k\!=\!1$. If the driver takes a rest at $S_k$ (i.e., $\tilde{b}_k\!=\!1$), $S_k$ becomes the start of a new consecutive driving period. Thus, the previous accumulated driving time $c_k$ will not be accounted for, as illustrated by the yellow lines in Fig.~\ref{Fig.2}(a). Otherwise (i.e., $\bar{b}_k\!=\!1$, $\tilde{b}_k\!=\!0$), the accumulated travel time upon arriving at $r_k$ will be taken into account when computing the accumulated travel time at $r_{k+1}$, and a round-way detour between $r_k$ and $S_k$ will be taken, as illustrated by the red and yellow lines in Fig.~\ref{Fig.2}(b). If the driver neither drive to nor rest at $S_k$ (i.e., $\tilde{b}_k\!=\!0,\, b_k\!=\!0$), the consecutive driving time $c_{k+1}$ is then shown in Fig.~\ref{Fig.2}(c).

\subsection{Constraints on the Problem}
In what follows, we introduce the constraints imposed on the charging strategy due to the battery dynamics, HoS regulations, and the delivery deadline.
\subsubsection{Battery Constraints}
Let $e_f$ be the energy of a truck with a full battery. Due to the capacity limitation, the total energy that the truck can be charged at $S_k$ is restricted by
\begin{equation}
\label{eq:charging_bound}
0\leq\Delta{e}_k\leq{e_f-\big(e_k-\bar{P}d_k\big)},\quad k=0,\dots, N\!-\!1,
\end{equation}
where $e_k-\bar{P}d_k$ is the remaining energy in the battery when the truck arrives at $S_k$. 

Furthermore, to ensure that there is sufficient energy for reaching $S_k$, each $e_k$ with $k=0,\dots, N$ shall fulfill
\begin{equation}
    \label{eq:energy_constraint}
e_k\geq{e_s+\bar{P}d_k},\; k=0,\dots,N\!-\!1,\quad 
e_N\geq{e_s},
\end{equation}
where $e_s$ denotes a constant safety margin. 

\subsubsection{HoS Regulations Constraints} The HoS regulations involve three quantities, namely, the maximum consecutive driving time, denoted as $T_d$, the maximum daily driving time, denoted as $\bar{T}_d$, and the minimum mandatory rest time before starting a new consecutive driving period, denoted as $T_r$. 

In line with the HoS regulations, the maximum consecutive driving time shall be bounded by $T_d$. That is, for $k\!=\!0,\dots,N$, the consecutive driving time $c_k$ is restricted by
\begin{equation}
\label{eq:Td_constraint}
c_k+d_k\leq T_d,\; k=0,\dots,N\!-\!1,\quad
c_N\leq T_d.
\end{equation}
Moreover, as the driver's daily driving time is no more than $\bar{T}_d$, we have that
\begin{equation}
    \label{eq:total_Td_constraint}
    \sum_{k=0}^{N}\tau_k+\sum_{k=0}^{N-1}2\bar{b}_kd_k\leq{\bar{T}_d},
\end{equation}
where, as defined above, $\bar{b}_k=b_k\!\lor\!\tilde{b}_k$, so that $\bar{b}_k\!=\!1$ if the truck visits $S_k$ and $0$ otherwise.

When charging at $S_k$, there is a preparation time $p_k$ before the battery can get charged. In addition, we consider staying at $S_k$ over $T_r$ as taking a rest. As a result, when $b_k\!=\!1$ and $\Tilde{b}_k\!=\!0$, the sum $t_k+p_k$ should be less than $T_r$, i.e., $t_k+p_k<T_r$. On the other hand, no such restriction is needed if $\Tilde{b}_k\!=\!1$. These constraints can be described compactly as
\begin{equation}
\label{eq:Tr_constraint}
    b_k\big(t_k+p_k\big)\leq\big(1-\tilde{b}_k\big)\big(T_r-\underline{\delta}\big)+\tilde{b}_k\overline{\delta},
\end{equation}
for $k=0,\dots,N-1$, where $\underline{\delta}$ is some small positive constant so that constraint $t_k+p_k\!<\!T_r$ is approximated by $t_k+p_k\leq T_r\!-\!\underline{\delta}$. The large constant $\overline{\delta}$ is introduced to approximate unboundedness above.

\subsubsection{Delivery Deadline Constraint}
Let the total time allowed to complete the trip be $\Delta T\!+\!\sum_{k=0}^{N}\tau_k$, where $\Delta T$ provides an upper bound on the extra time spent due to charging and rest. Then the constraint imposed by the deadline is 
\begin{equation}
    \label{eq:deadline}
    \sum_{k=0}^{N-1}\max\Big\{b_k\big(2d_k\!+\!p_k\!+\!t_k\big), \tilde{b}_k\big(2d_k\!+\!T_r\big)\Big\}\leq \Delta T.
\end{equation}

\section{Exact Solution to the Optimal Charging Problem with HoS Regulations}\label{Section III}
This section presents the optimization problem for determining the optimal charging strategy while fulfilling the HoS regulations. We start by introducing the optimal charging problem, followed by the exact solution and the computational complexity analysis of the problem.

\subsection{Optimal Charging Problem}
\subsubsection{Cost Function}
Our goal is to complete the delivery mission on time under the HoS regulations while saving operational costs. This includes the cost of charging and economic loss due to extra labor costs. Specifically, the expenses resulting from charging the truck at selected stations along its route are defined as
\begin{equation*}
    F_1\big(b_0,t_0,\dots,b_{N-1},t_{N-1}\big)=\sum_{k=0}^{N-1}\xi_{k}b_kt_k,
\end{equation*}
where $\xi_{k}$ represents the electricity price per charging time unit in accordance with the charging power at $S_k$, and $t_k$ is the charging time at $S_k$. 

In addition, the cost due to the extra travel time during the entire trip is represented as
\begin{align}
&F_2\big(b_0,\tilde{b}_0,t_0,\dots,b_{N-1},\tilde{b}_{N-1},t_{N-1}\big)\nonumber\\
=&\sum_{k=0}^{N-1}\max\Big\{b_k\big(2d_k\!+\!p_k\!+\!t_k\big), \tilde{b}_k\big(2d_k\!+\!T_r\big)\Big\}\varepsilon,\label{Eq.14}
\end{align}
where, as previously defined, $T_r$ represents the minimum mandatory rest time specified by the HoS regulations. The monetary loss per extra travel time unit is denoted by $\varepsilon$. 

The cost function of the optimal charging problem is then of the following form
\begin{align}
&F\big(b_0,\tilde{b}_0,t_0,\dots,b_{N-1},\tilde{b}_{N-1},t_{N-1}\big)\nonumber\\
= &F_1\big(b_0,t_0,\dots,b_{N-1},t_{N-1}\big)+\nonumber\\
&F_2\big(b_0,\tilde{b}_0,t_0,\dots,b_{N-1},\tilde{b}_{N-1},t_{N-1}\big),\label{Eq.15}
\end{align}
which includes the cost of charging and the cost of extra travel time for completing the delivery mission.

\subsubsection{Optimization Problem}
Based on the battery dynamics, consecutive driving times, HoS regulations, and delivery deadline constraints formulated in Section~\ref{Section II}, as well as the cost function given above, the optimal charging strategy can be obtained by solving the following optimization problem
\begin{align*}
\min_{\{(b_k,\Tilde{b}_k,t_k)\}_{k=0}^{N-1}}&\quad F\big(b_0,\tilde{b}_0,t_0,\dots,b_{N-1},\tilde{b}_{N-1},t_{N-1}\big)\\
    \mathrm{s.\,t.}&\quad\eqref{eq:control}-\eqref{eq:deadline},
\end{align*}
where \eqref{eq:control} defines the domains of the decision variables $b_k$, $\Tilde{b}_k$, and $t_k$, \eqref{eq:battery_energy} and \eqref{eq:charging} characterize the battery dynamics during driving and charging, and \eqref{eq:consecutive_t} describes the consecutive driving times upon arriving at each ramp. The constraints imposed by the battery capacity and its safety margin are \eqref{eq:charging_bound} and \eqref{eq:energy_constraint}. The HoS regulations are characterized by \eqref{eq:Td_constraint}-\eqref{eq:Tr_constraint}. The constraint related to the delivery deadline is \eqref{eq:deadline}. The sufficient conditions under which the problem is feasible are given in Appendix~\ref{app:e}.

Note that the proposed formulation is flexible to incorporate various modifications, such as taking the sum of $\Bar{b}_k$ as the cost function \eqref{Eq.15} for a sparse selection of the stations, or replacing the linear approximation of battery dynamics \eqref{eq:charging} to nonlinear ones. For simplicity, we focus on the present setting.
 
\subsection{Exact Solution}
The optimal charging problem formulated above is a mixed integer program with bilinear constraints. Thus, it cannot be directly addressed by many standard solvers. To obtain the exact solution to the problem, one could iterate over all possible combinations of integer variables. Since the integer variables $b_k$ and $\tilde{b}_k$ admit $4$ combinations at each station, i.e., $(0,0)$, $(0,1)$, $(1,0)$, $(1,1)$, there are in total $4^{N}$ charging and rest choices, where $N$ is the number of stations. Therefore, the exact solution requires solving $4^N$ linear programs, which leads to high computational demands and is not practical.

Note that the bilinear constraints can be transformed into linear ones so that the problem becomes a standard mixed integer linear program. We demonstrate this transformation in Appendix~\ref{app:f}. However, the exact solution to the transformed problem may still require an exponential number of iterations; see \cite[p.~480]{bertsimas1997introduction}. Moreover, if the linear approximation of charging in \eqref{eq:charging} is replaced by nonlinear functions of $t_k$, such transformations would become obsolete. 

To obtain tractable charging strategies, especially for long-haul trips with many candidate charging and rest stations, a rollout-based approximate solution to the optimal charging problem is proposed in the following section.

\section{Approximate Solution to the Optimal Charging Problem via Rollout}\label{Section IV}
In this section, we introduce the proposed rollout scheme for the optimization problem formulated in Section~III. We first describe a basic form of the method within the context of a general mixed integer program, which is modified from the methods introduced in \cite{bertsekas2005rollout} and \cite[Section~3.4]{bertsekas2020rollout}. It is followed by a variant of the scheme. Then we demonstrate how the basic form, as well as its variant, can be applied to obtain an approximate solution to the optimal charging problem. In the Appendix, we provide an orientation for the connection between the mixed integer program studied here and general optimal control problems where the rollout method is originally devised, leading to further insights into our method. Based on this connection, additional variants are introduced there as well.

\subsection{Rollout for Mixed Integer Program}
Let us consider the following mixed integer program:
    \begin{align}
	\min_{(u,v)}& \quad G(u,v)\quad \mathrm{s.\,t.}~(u,v)\in\overline{C}, \label{eq:mixed_int}
	\end{align}
where $u=(u_0,\dots,u_{N-1})$ is composed of discrete elements, with each element $u_k$ belonging to a finite discrete set $U_k$, i.e., $u_k\!\in\! U_k$, $k\!=\!0,\dots,N\!-\!1$, and $v\!\in\! \Re^m$ where $\Re^m$ is the $m$-dimensional Euclidean space. The function $G$ maps elements in $U\!\times\! \Re^m$ to real numbers with $U\!=\!U_0\times \dots\times U_{N-1}$, and $\overline{C}$ is a nonempty subset of $U\!\times\! \Re^m$. 

When favorable structures are absent, problem \eqref{eq:mixed_int} can be difficult to address. A naive approach is to enumerate all possible values of $u$, and then solve just as many optimization problems that involve only the continuous variable $v$. However, the number of such problems could increase exponentially as the dimension of $u$ increases. On the contrary, the number of continuous optimization problems involved in our scheme grows only linearly with $N$, as we will see shortly. 

For our proposed scheme to find a feasible solution in theory, we assume that there is a known $\Bar{u}\!=\!(\Bar{u}_0,\dots,\Bar{u}_{N-1})$, referred to as the \emph{base solution}, such that $(\Bar{u},\Bar{v})\in \overline{C}$ for some $\Bar{v}\!\in\! \Re^m$. In other words, if we define a set $C$ as
\begin{equation}
    \label{eq:feasible_u}
    C\!=\!\big\{u\!\in\! U\,|\,\text{$(u,v)\!\in\! \overline{C}$ for some $v\!\in\! \Re^m$}\big\},
\end{equation}
then our scheme relies on the assumption that some $\Bar{u}\in C$ is known. Based on this condition, the proposed method focuses on the discrete variables \emph{one at a time}. In particular, it first computes the $\tilde{u}_0$ via solving
\begin{equation}
    \label{eq:rollout_mixed_int_0}
    \begin{aligned}
    \Tilde{u}_0\!\in\!\arg\min_{u_0\in U_0}\min_{v\in \Re^m}& \quad G(u_0,\Bar{u}_1,\dots,\Bar{u}_{N-1},v)\\
	\mathrm{s.\,t.} &\quad (u_0,\Bar{u}_1,\dots,\Bar{u}_{N-1},v)\!\in\! \overline{C}.
    \end{aligned}    
\end{equation}
Having computed $\Tilde{u}_0$, it proceeds by solving
\begin{equation}
    \label{eq:rollout_mixed_int_1}
    \begin{aligned}
    \Tilde{u}_1\!\in\!\arg\min_{u_1\in U_1}\min_{v\in \Re^m}&~ G(\Tilde{u}_0,u_1,\Bar{u}_2,\dots,\Bar{u}_{N-1},v)\\
	\mathrm{s.\,t.} &~(\Tilde{u}_0,u_1,\Bar{u}_2,\dots,\Bar{u}_{N-1},v)\!\in\! \overline{C}.
    \end{aligned}    
\end{equation}
At last, it solves 
\begin{equation}
    \label{eq:rollout_mixed_int_end}
    \begin{aligned}
    \!\!\!\!\!\Tilde{u}_{N-1}\!\in\!\arg\min_{u_{N-1}\in U_{N-1}}\min_{v\in \Re^m}&~ G(\Tilde{u}_0,\dots, \Tilde{u}_{N-2},u_{N-1},v)\\
	\mathrm{s.\,t.} &~ (\Tilde{u}_0,\dots, \Tilde{u}_{N-2},u_{N-1},v)\!\in\!\overline{C}.
    \end{aligned}    
\end{equation}
Denoting as $\Tilde{u}$ the solution $(\Tilde{u}_0,\dots,\Tilde{u}_{N-1})$ computed above, referred to as the \emph{rollout solution}, the approximate solution obtained via our scheme is $(\Tilde{u},\Tilde{v})$ where 
\begin{equation}
    \label{eq:rollout_mixed_con}
    \Tilde{v}\in \arg\min_{(\Tilde{u},v)\in \overline{C}}G(\Tilde{u},v).
\end{equation}

We have the following result for the proposed scheme.
\begin{proposition}\label{prop:rollout_mixed_int}
Let $\Bar{u}\!\in\! C$ and consider $(\Tilde{u},\Tilde{v})$ obtained via \eqref{eq:rollout_mixed_int_0}-\eqref{eq:rollout_mixed_con}. We have that $(\Tilde{u},\Tilde{v})\!\in\! \overline{C}$ and
\begin{equation}
    \label{eq:rollout_bound}
    G(\Tilde{u},\Tilde{v})\leq \min_{(\Bar{u},v)\in \overline{C}}G(\Bar{u},v).
\end{equation}
\end{proposition}
\begin{proof}
See Appendix~\ref{app:b}. 
\end{proof}
\begin{remark}
Denote as $n$ the maximum number of elements contained in $U_{k}$. The naive scheme involves solving as many as $n^N$ continuous optimization problems, while our rollout scheme requires solving at most $nN$ such problems. Supposing that polynomial-time algorithms are used for continuous problems, our scheme can be executed in polynomial time. 
\end{remark}
\begin{remark}\label{rmk:infeasible}
The rollout scheme \eqref{eq:rollout_mixed_int_0}--\eqref{eq:rollout_mixed_int_end} can be carried out even starting from $\Bar{u}\!\not\in\! C$. In this case, the resulting rollout solution $\Tilde{u}$ may still be feasible.    
\end{remark}

\subsection{Variant of the Rollout Scheme}
The proposed scheme admits a few variants. Here we discuss one that is particularly relevant to our application. Additional variants are given in Appendices~\ref{app:c} and \ref{app:d}.

Suppose that $l$ different base solutions $\Bar{u}^1,\dots,\Bar{u}^\ell\!\in\!{U}$ are known. We can obtain their respective rollout solutions $\Tilde{u}^1,\dots,\Tilde{u}^\ell$ as well as the corresponding minimizing $\Tilde{v}^1,\dots,\Tilde{v}^\ell$. We then select $(\Tilde{u}^{i^*},\Tilde{v}^{i^*})$ where ${i^*}\!\in\! \arg\min_{i}\{G(\Tilde{u}^i,\Tilde{v}^i)\}_{i=1}^\ell$. Clearly, we have the following performance bound
\begin{equation}
    \label{eq:parallel_rollout_bound}
    G(\Tilde{u}^{i^*},\Tilde{v}^{i^*})\leq \min_{i\in \{1,\dots,\ell\}}\min_{(\Bar{u}^i,v)\in \overline{C}}G(\Bar{u}^i,v).
\end{equation}
However, this is at the expense of the increased computational demands, which are $\ell$-fold of that of the original scheme.

\subsection{Rollout-Based Charging Strategy}
In what follows, we show that the charging problem formulated in Section~\ref{Section III} belongs to the class of generic problem \eqref{eq:mixed_int}. As a result, the rollout scheme and its variant developed thus far can be applied to provide charging strategies. 

To this end, let us define as $u_k$ the pair $(b_k,\Tilde{b}_k)$, $k\!=\!0,\dots,N\!-\!1$, and $u\!=\!(u_0,\dots,u_{N-1})$. Accordingly, $U_k=\{(0,0),(1,0),(0,1),(1,1)\}$. We lump all the continuous variables involved in the charging problem as $v$, namely,
$$v\!=\!(t_0,e_0,\Delta{e}_0,c_0,\dots,t_{N-1},e_{N-1},\Delta{e}_{N-1},c_{N-1},e_N,c_N).$$
Then the function $F$ defined in \eqref{Eq.15} can be written as a function of $(u,v)$, which we denote as $G$. Moreover, let $\overline{C}$ denote the set of $(u,v)$ that fulfills the conditions \eqref{eq:control}-\eqref{eq:deadline}. Via the change of variables introduced here, the charging problem can be seen as an instance of the generic problem \eqref{eq:mixed_int}. 

To obtain charging strategies via the rollout scheme, we use two different base solutions. The first solution $\Bar{u}^1$ is referred to as the \emph{greedy solution}. Intuitively, the greedy solution sets $\Bar{u}_k^1\!=\!(1,1)$ if the battery energy $e_{k+1}$ upon arriving at $r_{k+1}$ does not fulfill constraint \eqref{eq:energy_constraint} without charging at $S_k$. Moreover, once $\Bar{u}^1_k\!=\!(1,1)$, the battery is fully charged at $S_k$. Another base solution $\Bar{u}^2$, referred to as the \emph{relaxed solution}, is obtained via solving a relaxation of the original problem, where the binary constraints are replaced by closed intervals $[0,1]$. If the optimal value for binary variables is nonzero, the relaxed solution sets respective binary variables to $1$. Apart from a base solution, the relaxation of the original problem also provides a lower bound of the optimal cost of the original problem. Together with the upper bounds \eqref{eq:rollout_bound} and \eqref{eq:parallel_rollout_bound}, we obtain a certificate for the optimality gap of the rollout scheme.

Note that either one of the two base solutions may not be feasible as they involve approximations of the original problem. Due to the presence of the HoS constraint \eqref{eq:total_Td_constraint} and the delivery deadline \eqref{eq:deadline}, computing a feasible base solution can be as hard as solving the original problem. On the other hand, owing to reasons discussed in Remark~\ref{rmk:infeasible}, both $\Bar{u}^1$ and $\Bar{u}^2$ are used in our simulation studies, and together they suffice for the practical needs. 
\section{Simulation Studies}\label{Section V}
\begin{figure*}[t]
\centering
\begin{minipage}{1\textwidth}
\centering
\subfigure[Swedish road network]
{\includegraphics[width=0.325\textwidth]{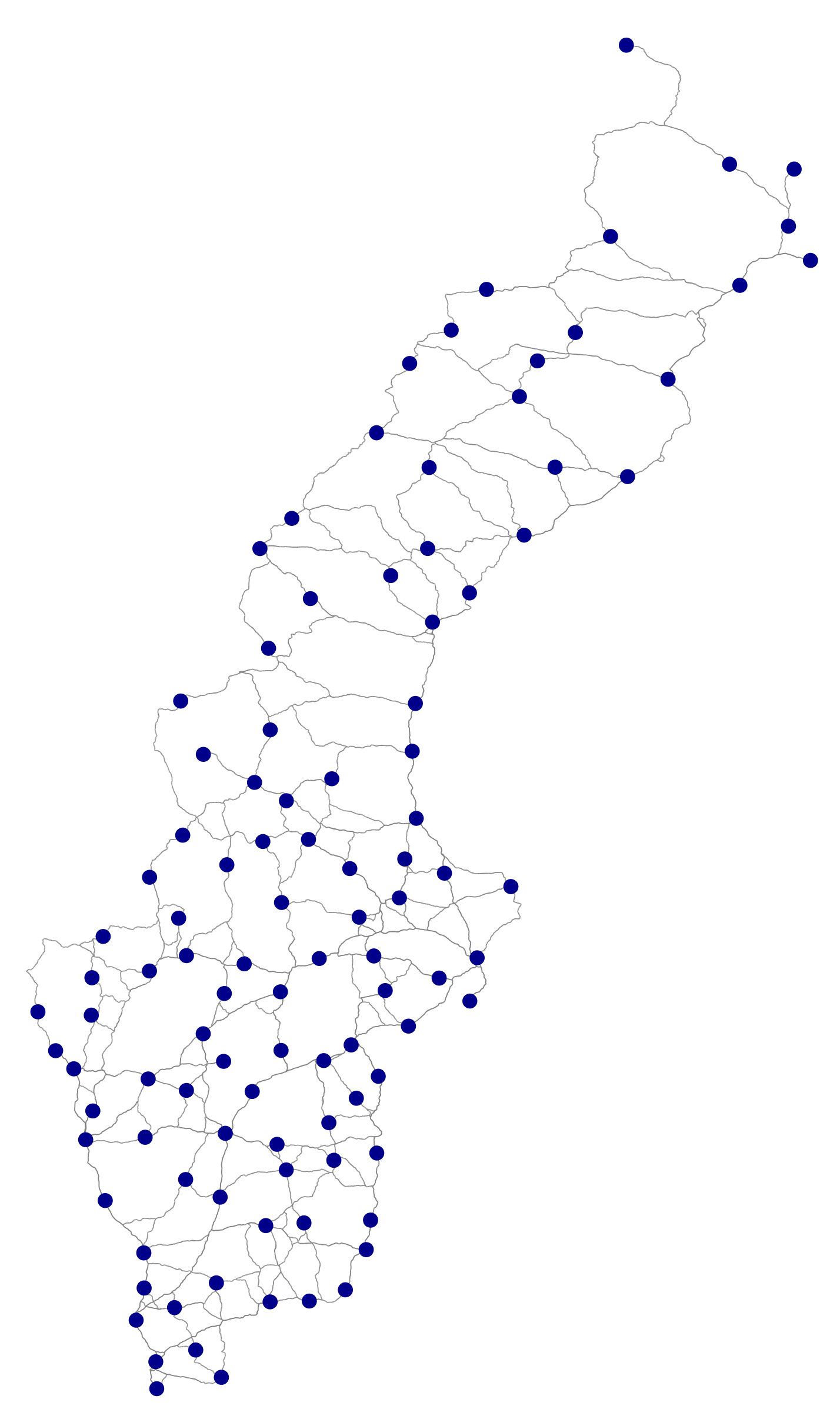}}
\subfigure[Potential charging and rest stations considered]
{\includegraphics[width=0.33\textwidth]{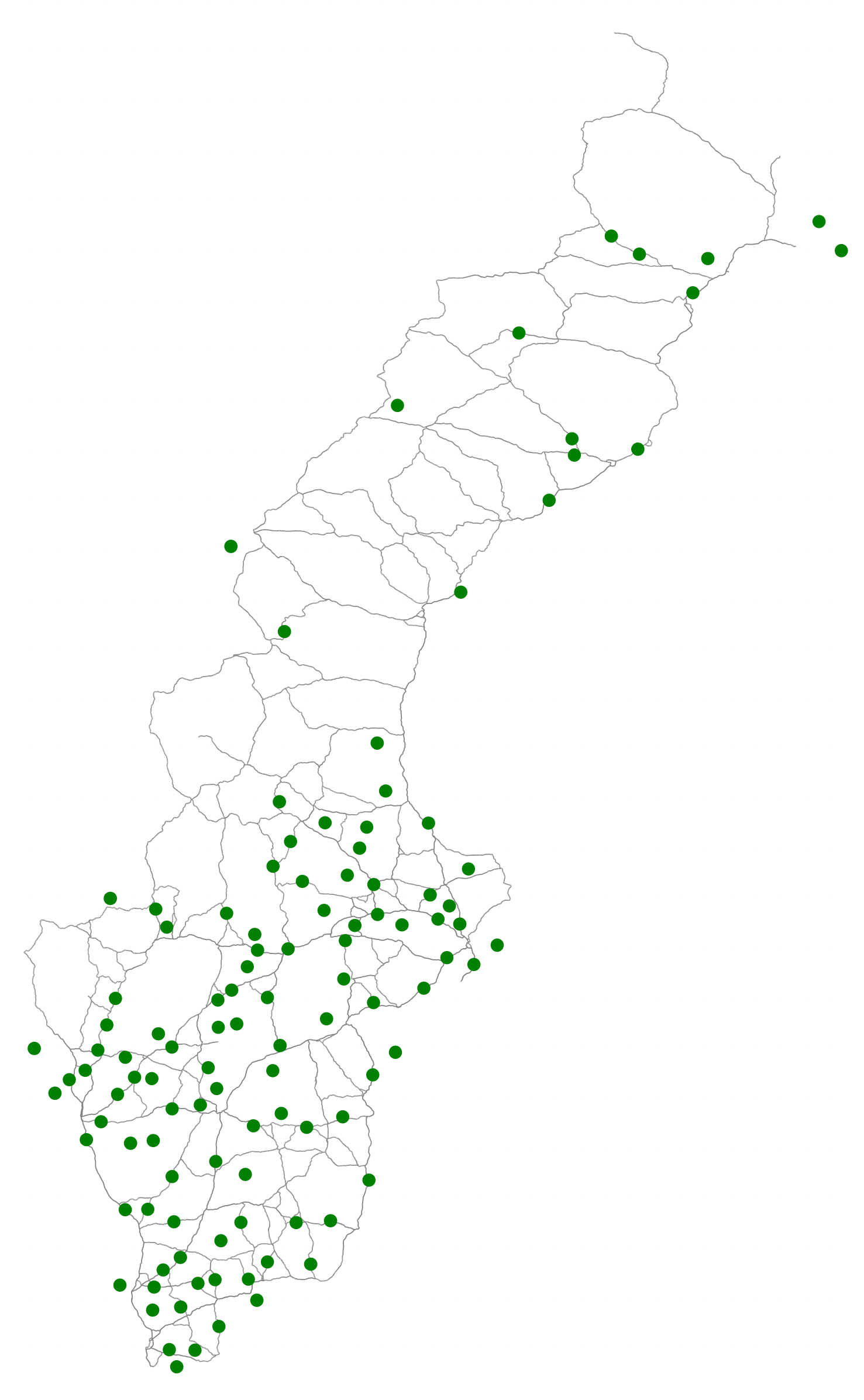}}
\subfigure[Route model of one truck]
{\includegraphics[width=0.33\textwidth]{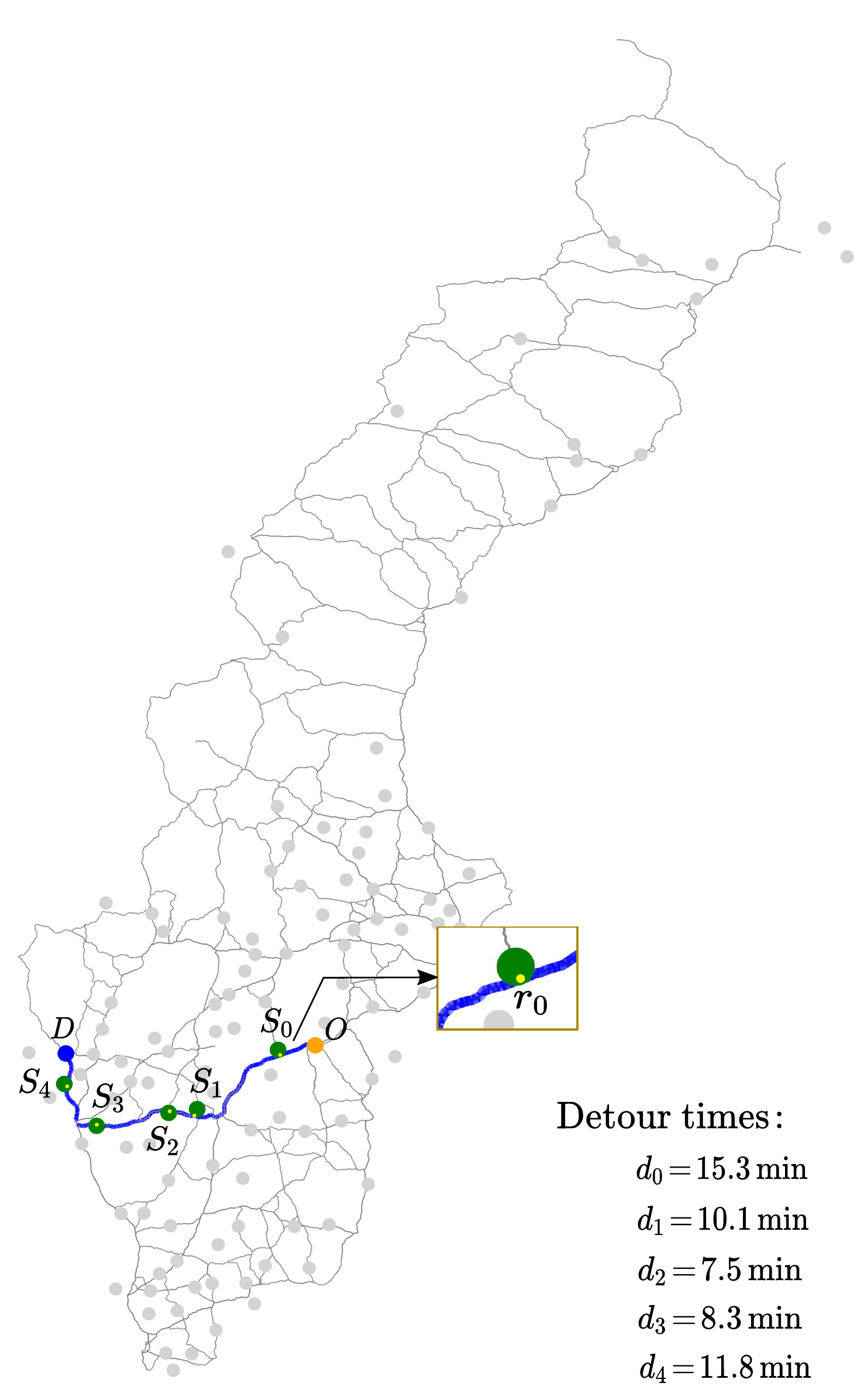}}
\vspace{-6pt}
\end{minipage}
\caption{(a) Swedish road network with $105$ road terminals, from which the OD pair of each delivery mission is selected. (b) The potential charging and rest stations considered are shown by the green nodes. (c) The transport route model of one truck, where $5$ charging and rest stations are available, and ramps leading to the shortest detours to stations are shown by the yellow nodes.}
\label{Fig.3}
\end{figure*}

\subsection{Setup}
\subsubsection{Transport Route}
We consider the Swedish road network with $105$ real road terminals, with each terminal aggregating the freight transportation demand within one region. The coordinates of the road terminals, as shown by the blue nodes in Fig.~\ref{Fig.3}(a), are obtained from the SAMGODS model~\cite{bergquist2016representation}, which is employed by public authorities to analyze and predict freight transport flows between different regions in Sweden. We generate the delivery missions for trucks by randomly selecting their origin and destination pairs (i.e., OD pairs) from the $105$ road terminals. As only a very few charging stations for electric trucks are in operation nowadays, the other real road terminals obtained from the SAMGODS model (except for those considered as origins and destinations) are used as potential charging and rest stations in our simulation, as shown by the green nodes in Fig.~\ref{Fig.3}(b). Given each OD pair related to a delivery mission, the shortest route between the OD pair is pre-planned and obtained from \textit{OpenStreetMap}\cite{OpenStreetMap}. Given a certain search range, the charging and rest stations along the route are identified. Accordingly, the travel times $\{\tau_k\}$ on each segment of the route and the detour times $\{d_k\}$ with $k\!=\!0,\dots, N\!-\!1$ are accessible from \textit{OpenStreetMap}. The route model of one truck is illustrated in Fig.~\ref{Fig.3}(c).  

The latest published data for electric trucks manufactured by Scania~\cite{ElectricTruck} is employed in setting the parameters. We consider electric trucks at a load capacity of $40$ tonnes with an installed battery capacity of $624$ kWh and a usable battery capacity of $468$ kWh, with up to $350$ kilometers driving range. The usable battery energy is $e_f-e_s$, which can be varied from $0$ to $468$ kWh. For safe operation purposes, $e_s$ is set as $25\%$ of the installed battery capacity. We assume that trucks drive at a constant speed of $82$ km/h, resulting in approximately $1.83$ kWh/min of battery consumption on the route. In addition, the electricity price for charging is considered as $0.36$ \texteuro/kWh, and the monetary loss per minute due to extra travel time is $0.4$ \texteuro, based on truck drivers' salaries per hour in Sweden in 2023. We apply the EU's HoS regulations nowadays, where $T_d$ is $4.5$ hours, $T_r$ is $45$ minutes, and $\bar{T}_d$ equals $9$ hours. For each trip, $\Delta{T}$ is considered as $150$ minutes. The values of other parameters are provided in Table~\ref{Table1}.

\subsubsection{Parameter Settings} 
\begin{table}[t]
\caption{Parameter Values} 
\vspace{-5pt}
\centering
\begin{tabular}{|c|c|c|c|c|} 
\hline
& & & & \\[-1.2ex]
  \raisebox{1.3ex}{$P_k$ [kW]}& \raisebox{1.3ex}{$P_{\max}$ [kW]}& \raisebox{1.3ex}{$e_f$ [kWh]}&\raisebox{1.3ex}{$\bar{P}$ [kWh/min]} &\raisebox{1.3ex}{$p_k$ [min]}
\\ [-0.5ex]
\hline 
& & & & \\[-0.6ex]
\raisebox{1.ex}{$300$} & \raisebox{1.ex}{$375$} & \raisebox{1.ex}{$624$} & \raisebox{1.ex}{$1.83$} & \raisebox{1.ex}{$6$}\\[-0.1ex]
\hline
\end{tabular}
\label{Table1}
\end{table}

\subsection{Solution Evaluation}
\begin{figure}[t]
     \centering
     \includegraphics[width=0.985\linewidth]{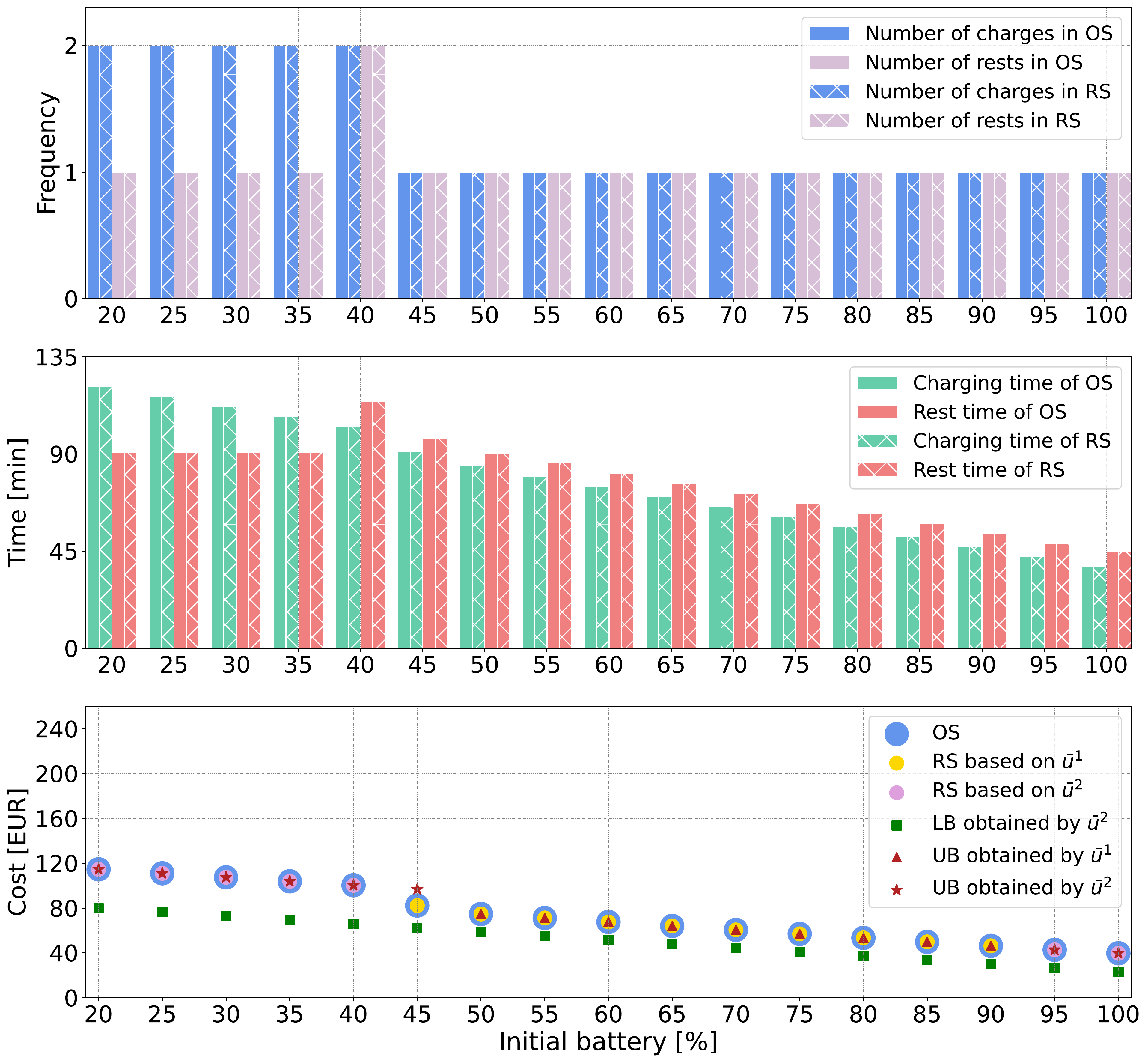}
     \vspace{-3pt}
      \caption{Comparison results of scenario 1 ($N\!=\!5$).}
      \label{Fig.4}
\end{figure}

\begin{figure}[t]
    \centering
    \includegraphics[width=0.985\linewidth]{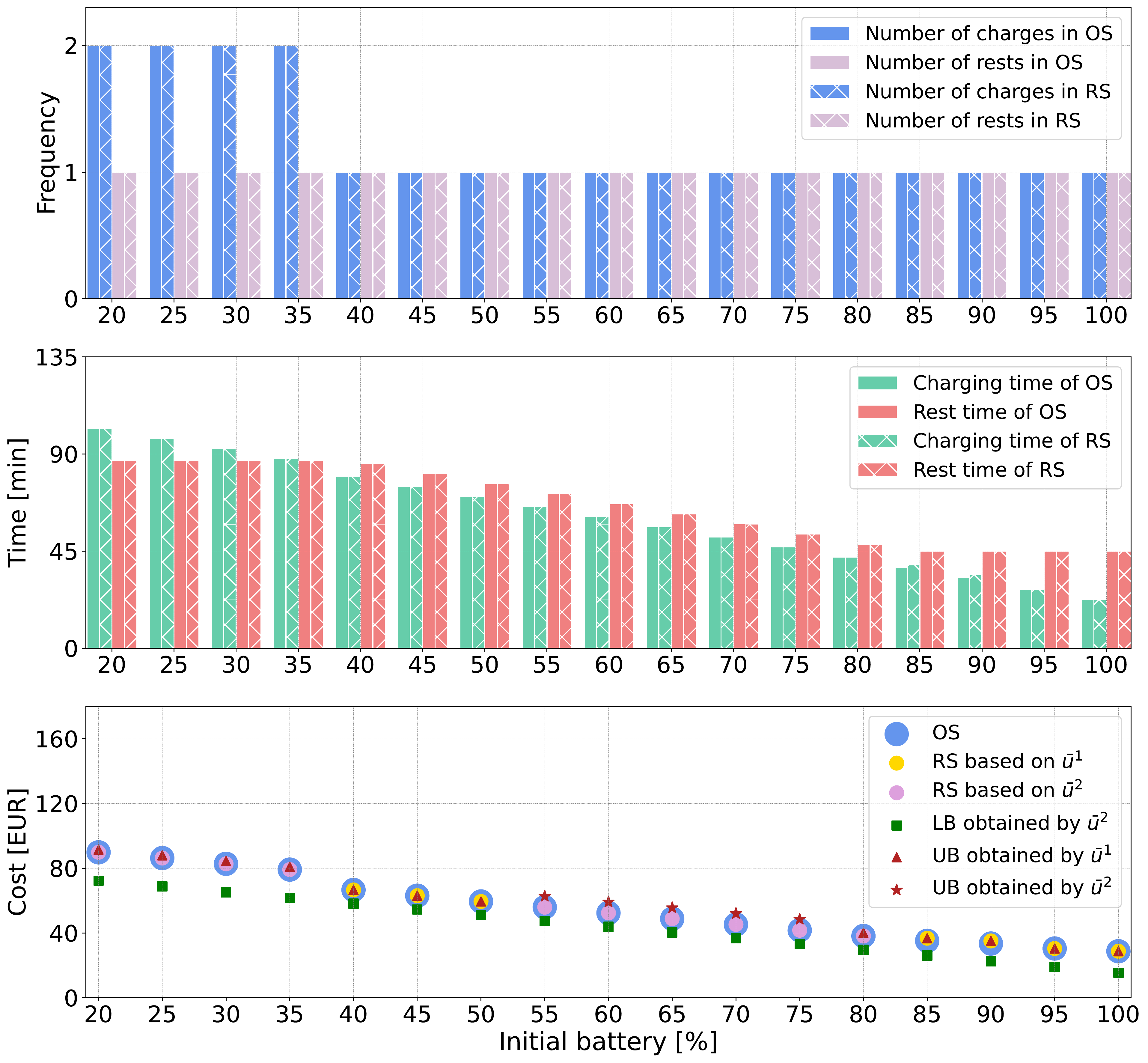}
    \vspace{-3pt}
     \caption{Comparison results of scenario 2 ($N\!=\!6$).}
     \label{Fig.5}
\end{figure}

\begin{figure}[t]
   \centering
   \includegraphics[width=0.985\linewidth]{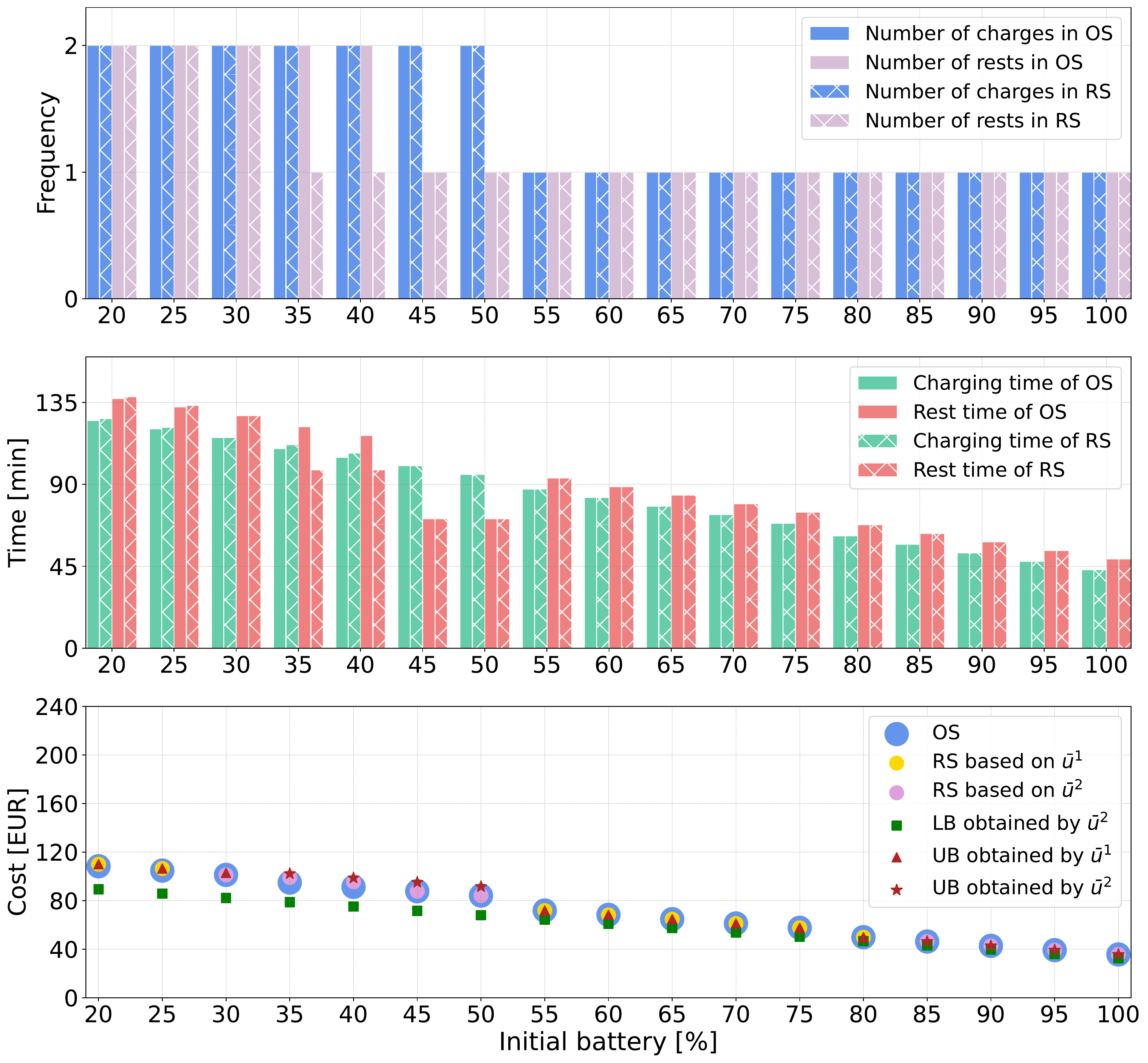}
   \vspace{-3pt}
   \caption{Comparison results of scenario 3 ($N\!=\!7$).}
   \label{Fig.6}
\end{figure}

\begin{figure}[t]
   \centering
\includegraphics[width=0.983\linewidth]{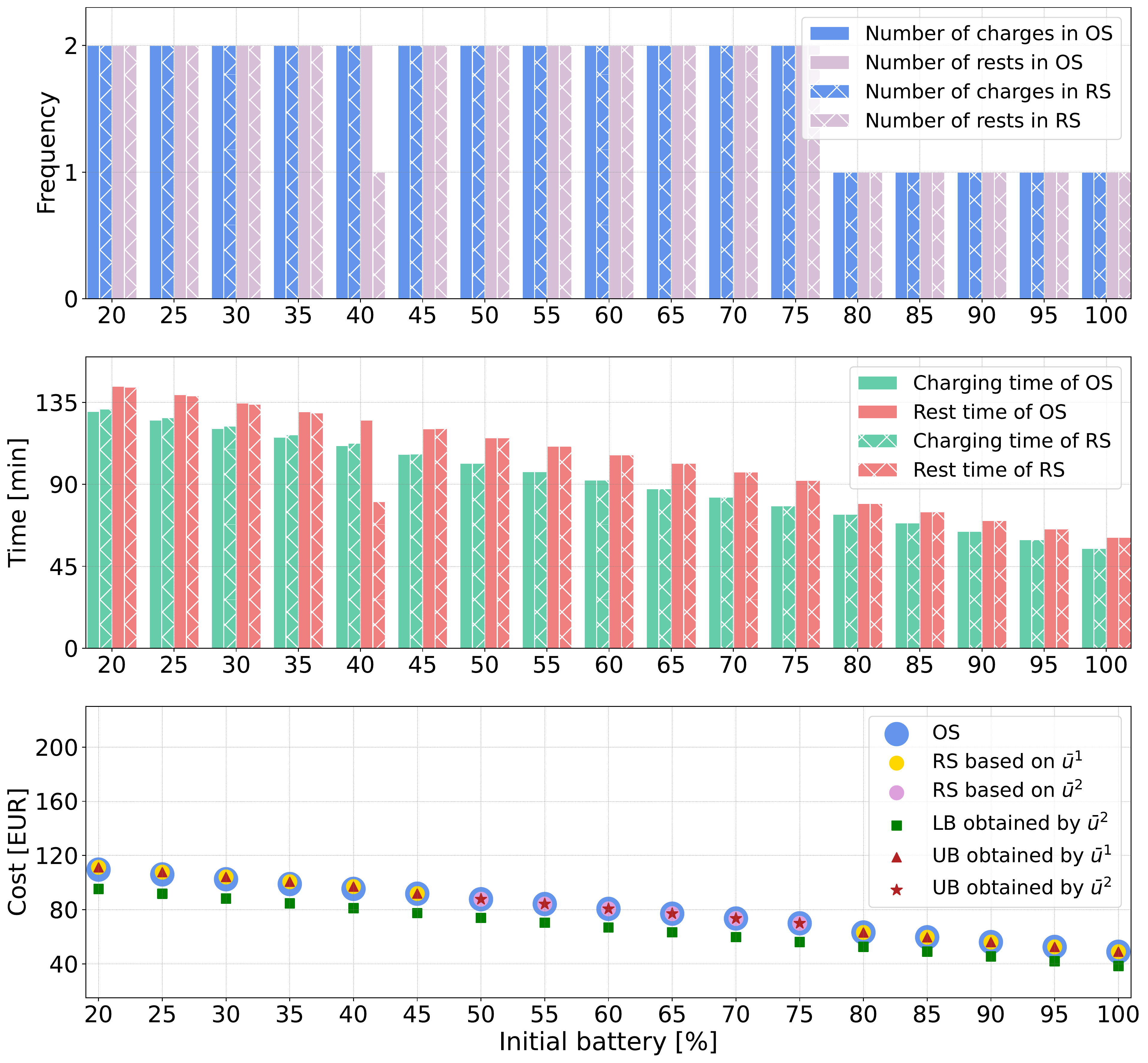}
   \vspace{-3pt}
   \caption{Comparison results in scenario 4 ($N\!=\!8$).}
   \label{Fig.7}
\end{figure}

\begin{figure}[t]
   \centering
\includegraphics[width=0.985\linewidth]{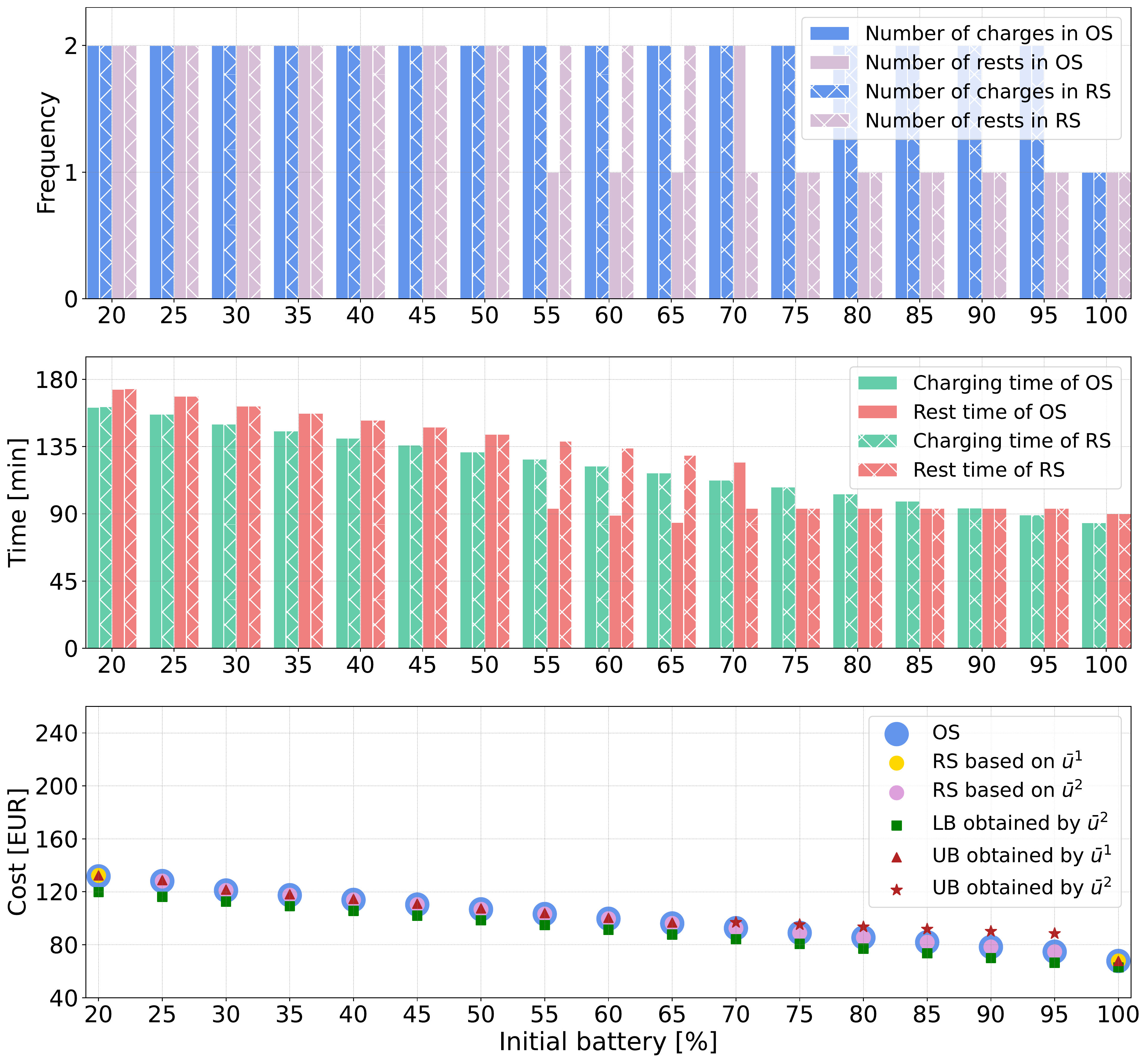}
   \vspace{-3pt}
   \caption{Comparison results in scenario 5 ($N\!=\!9$).}
   \label{Fig.8}
\end{figure}

\begin{figure}[t]
   \centering
   \includegraphics[width=0.982\linewidth]{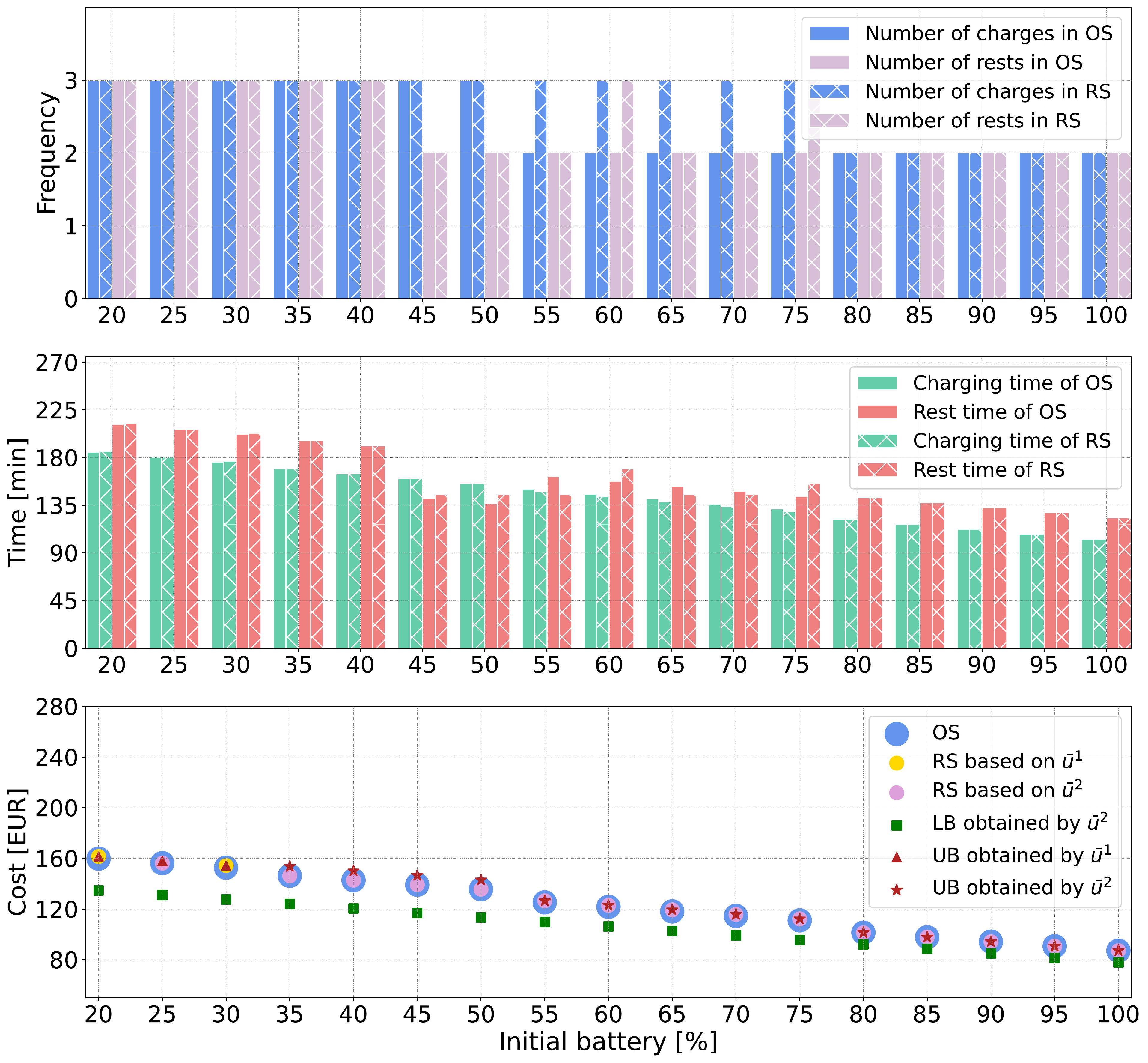}
   \vspace{-3pt}
   \caption{Comparison results in scenario 6 ($N\!=\!10$).}
   \label{Fig.9}
\end{figure}

To evaluate the rollout-based charging strategy, we conduct simulation studies for trucks in $6$ scenarios where $N$ is varied from $5$ to $10$, and in each scenario, the proportion of the initial battery is changed from $20\%$ to $100\%$, incremented with $5\%$. The optimal solution is computed by enumerating all the combinations of the binary variables, and the rollout solution is obtained by taking $\bar{u}^1$ and $\bar{u}^2$ as the base solutions. Both solutions use \textit{Gurobi} as the linear program solver. For brevity, we refer to the optimal and rollout-based solutions as OS and RS, respectively, and refer to the lower and upper bounds of the optimal cost of the rollout solution as LB and UB. The number of charges and rests, charging and rest times, as well as the total costs of the entire trip, compared between the OS and RS in each scenario, are given in Fig.~\ref{Fig.4}-Fig.~\ref{Fig.9}. The code for one sample is provided at\footnote{See \href{https://yuchaotaigu.github.io/research/CDC23.ipynb}{https://yuchaotaigu.github.io/research/CDC23.ipynb} for a sample implementation.}. Here, we note that $\Delta{T}$ is relaxed to $220$ minutes for the scenarios $N\!\!=\!9,10$ so that there exists a feasible solution for the problem. As we can see from Fig.~\ref{Fig.4}-Fig.~\ref{Fig.9}, the rollout-based charging and rest decisions are near-optimal in the majority of the cases, resulting in similar charging time, rest time, and total costs in comparison with the optimal solutions. In addition, the UBs and LBs in these figures indicate that the proposed greedy and relaxed base solutions provide good performance bounds for the rollout solution.  

\begin{table}[t]
\caption{Comparison between the OS and RS} 
\vspace{-2pt}
\centering 
\begin{tabular}{|l|c|c|c|c|c|c|} 
\hline
& & & & & &\\[-1.2ex]
  \raisebox{1.2ex}{\text{$N$}\!\!}& \raisebox{1.2ex}{\textbf{$5$}}&\raisebox{1.2ex}{$6$}&\raisebox{1.2ex}{$7$}&\raisebox{1.2ex}{$8$} &\raisebox{1.2ex}{$9$} &\raisebox{1.2ex}{$10$}
\\ [-0.9ex]
\hline 
& & & & & &\\[-0.9ex]
\raisebox{1.ex}{\text{AOG-RS [$\%$]\!\!}} & \raisebox{1.ex}{$0$} &\raisebox{1.ex}{$0.55$}&\raisebox{1.ex}{$0.72$}&\raisebox{1.ex}{\!$0.49$\!} & \raisebox{1.ex}{\!$0.03$\!} & \raisebox{1.ex}{\!$0.42$\!}\\[-0.4ex]
\hline 
& & & & & &\\[-0.9ex]
\raisebox{1.ex}{\text{AOG-UB [$\%$]}\!} & \raisebox{1.ex}{$1.04$} &\raisebox{1.ex}{$5.46$}&\raisebox{1.ex}{$2.23$}&\raisebox{1.ex}{\!$0.49$\!} & \raisebox{1.1ex}{\!$4.28$\!} & \raisebox{1.ex}{\!$1.72$\!}\\[-0.4ex]
\hline 
& & & & & &\\[-0.9ex]
\raisebox{1.ex}{\text{ACT of RS [s]}\!} & \raisebox{1.ex}{$0.34$} &\raisebox{1.ex}{$0.42$}&\raisebox{1.1ex}{$0.57$}&\raisebox{1.ex}{\!$0.65$\!} & \raisebox{1.ex}{\!$0.84$\!} & \raisebox{1.ex}{\!$1.43$\!}
\\[-0.5ex]
\hline
& & & & & &\\[-0.85ex]
\raisebox{1.ex}{\text{ACT of OS [min]}} & \raisebox{1.ex}{$0.32$} &\raisebox{1.ex}{$1.34$}&\raisebox{1.ex}{$5.45$}&\raisebox{1.ex}{\!$24.02$\!} & \raisebox{1.1ex}{\!$98.50$\!} & \raisebox{1.ex}{\!$413.68$\!}\\[-0.1ex]
\hline
\end{tabular}
\label{Table2}
\end{table}

The optimality gap between RSs and OSs, UBs and OSs, and the computational efficiency of the RS and OS methods are shown in Table~\ref{Table2}. For each $N$ with a given initial battery, the optimality gap between the RS and OS is computed by $100\!\times\!(F(\text{RS})\!-\!F(\text{OS}))/F(\text{OS})$, where $F$ is the cost function defined by \eqref{Eq.15}. Similarly, the optimality gap between the UB and OS is computed by $100\!\times\!(\text{UB}\!-\!F(\text{OS}))/F(\text{OS})$. We show in Table~\ref{Table2} the average optimality gap (AOG) of $17$ situations for each $N$ and the average computational times (ACT) to obtain the RS and OS. It can be seen from Table~\ref{Table2} that the computational demands for obtaining OSs increase exponentially with the increase in $N$. By employing the proposed RS scheme, the computational time decreases significantly, taking less than $2$ seconds, while having an average optimality gap within $1\%$, which illustrates the desirable properties of our method.

\begin{figure}[t]
     \centering
     \includegraphics[width=0.9\linewidth]{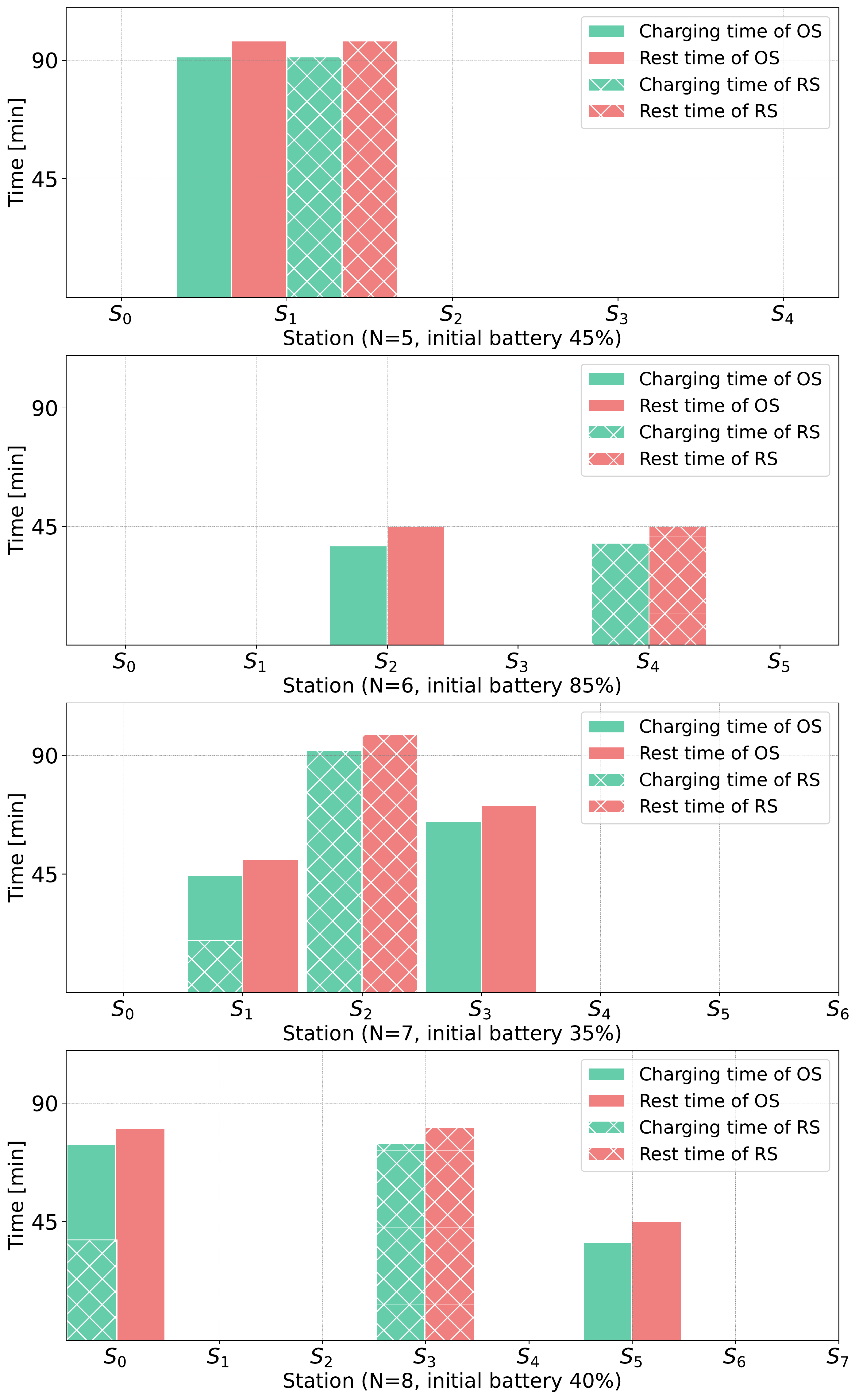}
     \vspace{-3pt}
      \caption{The charging and rest times in OSs and RSs in $4$ selected scenarios. For $N\!=\!5$ with the initial battery being $45\%$, the two solutions are identical.}
      \label{Fig.10}
   \end{figure}

\begin{figure}[t]
   \centering
   \includegraphics[width=0.95\linewidth]{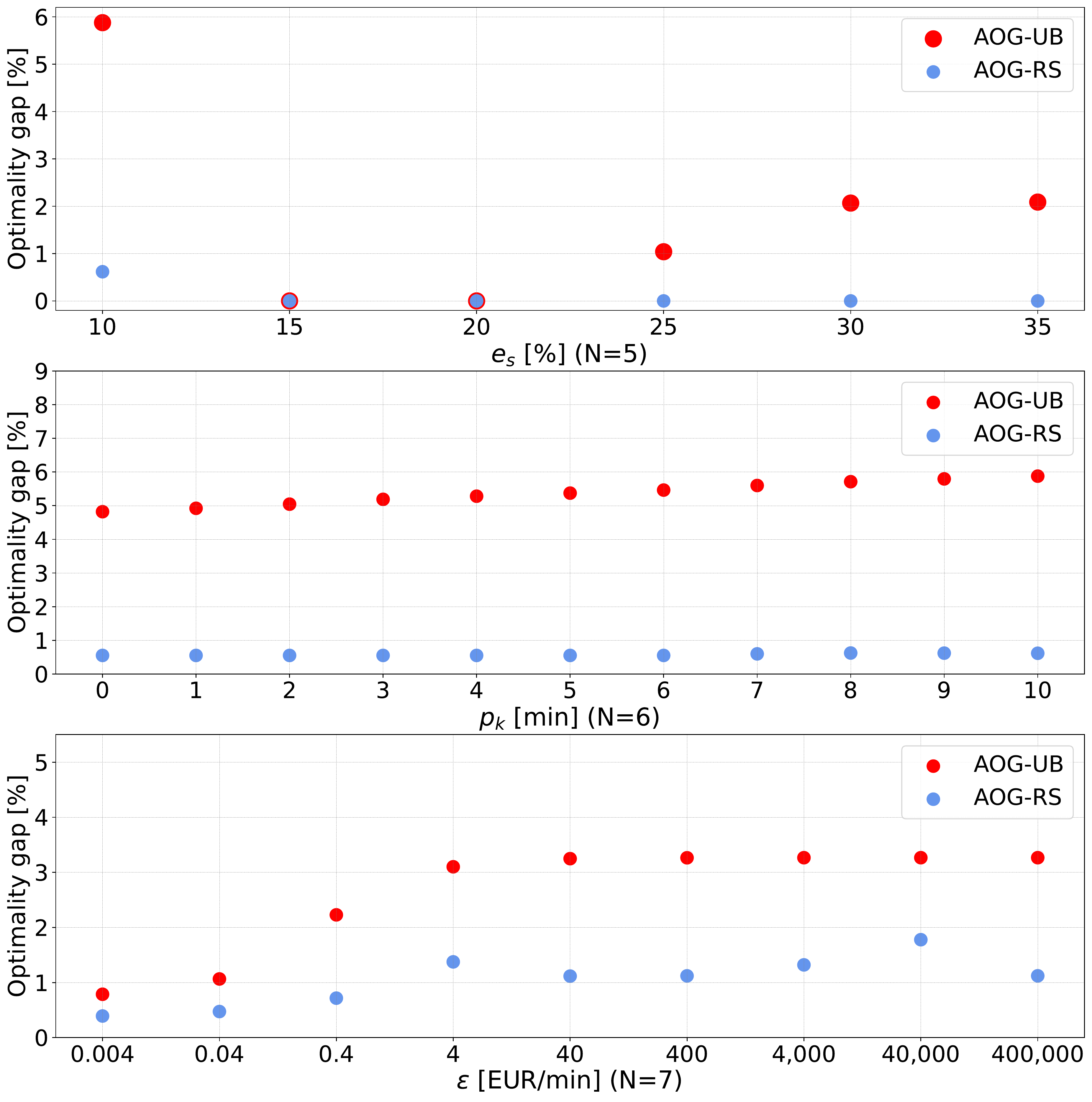}
   \vspace{-6pt}
   \caption{Sensitivity study of the rollout method to parameters $e_s$, $p_k$, and $\epsilon$. It can be seen that the performance of the rollout solution remains satisfactory and consistent.}
   \label{Fig.11}
\end{figure}
   
Fig.~\ref{Fig.10} shows the charging and rest times in OSs and RSs for the selected scenarios $N\!=\!5,6,7,8$ with the initial battery in each scenario being given in the figure. The results show that although the RS might provide trucks with different stations for charging and rest in contrast to the OS, it results in a near-optimal charging and rest time in total to fulfill the delivery mission while meeting the HoS regulations. 

Moreover, to evaluate how parameter selection affects the rollout solution performance, a parameter sensitivity study is performed, where a wide range of values are assigned to $e_s$, $p_k$, and $\epsilon$ for $N\!\!=\!5,6,7$ and still, the initial battery in the truck is varied from $20\%$ to $100\%$, incremented by $5\%$ for each parameter setting. The AOG between the UB and OS and that between the RS and OS are shown in Fig.~\ref{Fig.11}. The parameter sensitivity study shows that the proposed rollout scheme adapts to different parameter selections and remains reliable with small optimality gaps.
\vspace{5pt}

\section{Conclusion}\label{Section VI}
This letter investigated the optimal charging strategy for electric trucks, which allows freight drivers to determine where and how long to recharge trucks to complete the delivery task before deadlines while respecting the HoS regulations. We assumed that every truck has a pre-planned route with a given collection of charging and rest stations. The optimal charging problem of each truck was modeled as a mixed integer program integrated with bilinear constraints, which is computationally intractable to be solved exactly. As an approximate scheme, a rollout-based charging strategy was proposed, which provides near-optimal solutions to the problem with solid performance guarantees while reducing the computational load drastically. Compared to the existing literature, our modeling method allows for handling the HoS regulations subject to delivery deadlines. Moreover, the rollout-based solution of high efficiency is promising to be applied in real-time strategy planning to cope with travel time uncertainties. Future work could be developing optimal charging strategies for electric trucks with limited charging resources at stations.

\appendices
\appendix
In the appendices, we first give a brief introduction to the optimal control problem in Appendix~\ref{app:a}. It is within this context that the rollout scheme, as well as Newton's step interpretation of the method, is developed. Then we show in Appendix~\ref{app:b} how the mixed integer program considered here can be transformed as an equivalent optimal control problem involving only discrete variables, thus proving Prop.~\ref{prop:rollout_mixed_int}. Next, we provide additional variants of the proposed scheme and connect it to a classical method within the optimal control context in Appendices~\ref{app:c} and \ref{app:d}. The majority of the proof arguments and transformation are modified from \cite{bertsekas2005rollout} and \cite[Section~3.4]{bertsekas2020rollout}. In Appendices~\ref{app:e} and \ref{app:f}, we provide further details on our optimal control problem, including the sufficient conditions for its feasibility and a procedure through which the bilinear constraints of our problem can be linearized.

\subsection{Optimal Control Problems and Rollout}\label{app:a}
We consider optimal control problems involving dynamics 
$$x_{k+1}=f_k(x_k,u_k),\quad k=0,\dots,N\!-\!1,$$
where $x_k$ and $u_k$ are the state and control at time $k$, which belong to some sets $X_k$ and $U_k$ that contain finite elements, respectively, and $f_k$ is some function. Each control must be chosen from a finite constraint set $U_k(x_k)$ that depends on the current state $x_k$. We are interested in the policy $\pi\!=\!\{\mu_0,\mu_1,\dots,\mu_N\}$, which is a sequence of functions $\mu_k:X_k\mapsto U_k$ such that $\mu_k(x_k)\in U_k(x_k)$ for all $x_k$. The set of all the policies is denoted as $\Pi$. 

When at $x_k$ and applying $u_k\!\in\! U_k(x_k)$, there is a stage cost $g_k(x_k,u_k)\!\in\! \Re\cup\{\infty\}$. In addition, there is a terminal cost $g_N(x_N)\in \Re\cup\{\infty\}$ for being $x_N$ at $N$th stage. For a policy $\pi$, starting from $x_k$, the total cost accumulated over $N-k$ stages plus the terminal cost are denoted as $J_{k,\pi}(x_0)$, i.e.,
$$J_{k,\pi}(x_k)=g_N(x_N)+\sum_{i=k}^{N-1}g_i\big(x_i,\mu_i(x_i)\big),$$
where $x_{i+1}\!=\!f_i\big(x_i,\mu_i(x_i)\big)$ for $i\!=\!k,\dots,N\!-\!1$. For brevity, we denote $J_{0,\pi}$ as $J_\pi$. Within this context, one hopes to find the optimal cost $J^*$ such that
$$J^*(x_0)=\min_{\pi\in \Pi}J_{\pi}(x_0).$$

For many problems, computing the optimal policy is impractical. In those cases, rollout is a simple yet reliable approximate solution method. Based upon a known policy $\Bar{\pi}=\{\Bar{\mu}_0,\dots,\Bar{\mu}_{N-1}\}$, referred to as the \emph{base policy}, rollout scheme computes a new policy $\Tilde{\pi}=\{\Tilde{\mu}_0,\dots,\Tilde{\mu}_{N-1}\}$, referred to as the \emph{rollout policy}, through computations
\begin{equation}
    \label{eq:rollout_op}
    \Tilde{\mu}_k(x_k)\in\arg\min_{u_k\in U_k(x_k)}\Big\{g_k(x_k,u_k)+J_{\Bar{\pi},k+1}\big(f_k(x_k,u_k)\big)\Big\},
\end{equation}
where $x_{k+1}\!=\!f_{k}\big(x_k,\Tilde{\mu}_k(x_k)\big)$, $k\!=\!0,\dots,N\!-\!1$. The rollout policy computed above is no worse than the base policy, as is given in the following proposition, which is adopted from \cite[Prop.~3.3.1]{bertsekas2020rollout}.
\begin{proposition}\label{prop:rollout_original}
Let $\Bar{\pi}\!\in\!\Pi$ and consider $\Tilde{\pi}$ obtained via \eqref{eq:rollout_op}. Then we have that
\begin{equation}
    \label{eq:rollout_bound_general}
    J_{k,\Tilde{\pi}}(x_k)\leq J_{k,\Bar{\pi}}(x_k),
\end{equation}
for $k\!=\!0,1,\dots,N\!-\!1$ and $x_k\!\in\! X_k$.
\end{proposition}

In fact, rollout can be interpreted as one step of Newton's method for computing $J^*(x_0)$ with its starting point provided by the base policy. Extensive discussions on rollout that go well beyond the context considered here can be found in \cite{bertsekas2020rollout,bertsekas2022lessons}. 

\subsection{Mixed Integer Problem and Rollout}\label{app:b}
We now show that the mixed integer problem formulated in \eqref{eq:mixed_int} can be transformed into the equivalent optimal control problem discussed above, and the algorithm described in Section~\ref{Section IV} is the rollout method applied to this equivalent problem. Based upon this transformation and Prop.~\ref{prop:rollout_original}, Prop.~\ref{prop:rollout_mixed_int} can be proved.

\begin{proof}[Proof of Prop.~\ref{prop:rollout_mixed_int}]
We define a fictitious state `null' as the initial state $x_0$ so that $X_0$ is the singleton $\{\text{`null'}\}$. The set $X_1$ is defined as $U_0$, and for $k\!=\!1,\dots,N\!-\!1$, the sets $X_{k+1}$ are defined recursively as $X_{k+1}\!=\!X_k\!\times\! U_k$. As a result, the set $X_N$ containing all terminal states equals $U$. The control constraint sets are independent of states, i.e., $U_k(x_k)=U_k$ for all $x_k\in X_k$ and $k$.

Given the current state $x_k$ and the control $u_k\in U_k(x_k)$, the dynamics $f_k$ takes the form
$$x_{k+1}=(x_k,u_k),\quad k=0,1,\dots,N\!-\!1;$$
namely concatenating the control $u_k$ to the current state $x_k$. The stage costs associated with all the state-control pair $(x_k,u_k)$ are $g_k(x_k,u_k)\equiv 0$, while the terminal cost is
$$g_N(x_N)=\min_{(x_N,v)\in \overline{C}}G(x_N,v).$$
In particular, $g_N(x_N)\!=\!\infty$ for $x_N\!\not\in\!C$. Therefore, given some $u\!=\!(u_0,\dots,u_{N-1})\!\in\!U$, we can implicitly define a policy $\pi\!=\!\{\mu_0,\dots,\mu_{N-1}\}$ so that $\mu_k(x_k)\equiv u_k$ and
$$J_\pi(x_0)=\min_{(u,v)\in \overline{C}}G(u,v).$$

With the equivalent optimal control problem in mind, we can see that the sequence of minimization \eqref{eq:rollout_mixed_int_0}-\eqref{eq:rollout_mixed_int_end} is the computations \eqref{eq:rollout_op} for $k\!=\!0,\dots,N\!-\!1$, and the bound \eqref{eq:rollout_bound} in Prop.~\ref{prop:rollout_mixed_int} is equivalent to \eqref{eq:rollout_bound_general}. In particular, $\Bar{u}\!\in\! C$ implies that $J_{\Bar{\pi}}(x_0)\!<\!\infty$. As a result, $J_{\Tilde{\pi}}(x_0)\!\leq\! J_{\Bar{\pi}}(x_0)$, which means that $\Tilde{u}\in C$.
\end{proof}

Based upon the transformation introduced above, the other variant of the proposed scheme discussed in Section~\ref{Section IV} can be interpreted accordingly within the context of the equivalent optimal control problem. 

\subsection{Variant Based on On-Line Policy Iteration}\label{app:c}
The rollout scheme can be repeated to further enhance the performance. In particular, given a policy $\Tilde{\pi}^0$, we can obtain a new policy $\Tilde{\pi}^1\!=\!\{\Tilde{\mu}_0^1,\dots,\Tilde{\mu}_{N-1}^1\}$ through computations similar to \eqref{eq:rollout_op} with $\Tilde{\pi}^0$ and $\Tilde{\mu}_k^1$ in place of $\Bar{\pi}$ and $\Tilde{\mu}_k$. After obtaining $\Tilde{\pi}^i$, we may proceed to compute $\Tilde{\pi}^{i+1}$ in a similar manner. This scheme can be considered as the on-line policy iteration algorithm \cite{bertsekas2021line} adapted to the optimal control problem of concern. If computational resource permits, the policies obtained would converge in the sense that for some finite $\Bar{k}$, starting from the same $x_0$, the trajectories generated under $\Tilde{\pi}^{\Bar{k}}$
and $\Tilde{\pi}^{\Bar{k}+1}$ are identical. Note that the obtained policy $\Tilde{\pi}^{\Bar{k}}$ upon convergence need not be optimal. Instead, it is optimal for a modified problem; see \cite[Definition~2.1]{bertsekas2021line}.

For the mixed integer program \eqref{eq:mixed_int} considered here, given a base solution $\Tilde{u}^0$, we may compute its corresponding rollout solution $\Tilde{u}^1$ through the sequence of minimization \eqref{eq:rollout_mixed_int_0}-\eqref{eq:rollout_mixed_int_end} with $\Tilde{u}^0$ and $\Tilde{u}^1$ in place of $\Bar{u}$ and $\Tilde{u}$, which is equivalent to the rollout scheme applied to its equivalent optimal control problem. Then the repeated application of rollout also applies. In particular, treating the current solution $\Tilde{u}^i$ as the base solution, we can obtain the corresponding rollout solution $\Tilde{u}^{i+1}$ through similar computations.  

\subsection{Additional Variants of the Rollout Scheme}\label{app:d}

First, from the description of the scheme, it is clear that we can change the order in which the elements of $u$ are optimized. The validity of our scheme, as well as the corresponding performance guarantees stated in Prop.~\ref{prop:rollout_mixed_int} remains intact. For example, we can reverse the order and start by computing $\Tilde{u}_{N-1}$ with other elements fixed at $\Bar{u}_k$, $k=0,\dots,N-2$, and proceed backward.

Moreover, assume that after obtaining $\Tilde{u}_{k-1}$, the computational budget runs out. Then the tentative best solution $(\Tilde{u}_0,\dots,\Tilde{u}_{k-1},\Bar{u}_k,\dots,\Bar{u}_{N-1})$, which we denote as $\hat{u}^{k-1}$, is feasible in the sense that $\hat{u}^{k-1}\in U$. In addition, for the corresponding optimizer denoted as $\hat{v}^{k-1}$, we have the performance bound
$$G(\Tilde{u},\Tilde{v})\leq G(\hat{u}^{k-1},\hat{v}^{k-1})\leq \min_{(\Bar{u},v)\in \overline{C}}G(\Bar{u},v).$$
Therefore, our scheme has the character of an anytime algorithm. However, in the case where $k\!<\!N$, Newton's step interpretation described in Section~\ref{Section I} is not valid anymore.

\subsection{Sufficient Conditions for the Feasibility of Charging Problem}\label{app:e}
We provide a set of conditions under which a feasible solution can be obtained analytically. Some of the conditions discussed here are restrictive. Still, they may be taken as a starting point for the construction of the base solution.

For the battery parameters, it is natural to assume that 
\begin{align*}
e_{\text{ini}}\geq& e_s+\Bar{P}(\tau_0+d_0),\\
e_f\geq& e_s+\Bar{P}(d_{k-1}+\tau_k+d_k),~\;k=1,\dots,N\!-\!1,\\
e_f\geq& e_s+\Bar{P}\tau_N,
\end{align*}
which means that the initial energy suffices for the trip to the first station, and the fully charged battery can cover the trip connecting two stations.

Similarly, we may expect that the duration of the trips connecting two stations is less than $T_d$, i.e.,
\begin{align*}
\tau_0+d_0\leq& T_d,\\
d_{k-1}+\tau_k+d_k\leq& T_d,~\;k=1,\dots,N\!-\!1,\\
d_{N-1}+\tau_N\leq& T_d.
\end{align*}
When there are limited numbers of $S_k$, the total driving time involving all detours can be no more than $\Bar{T}_d$, i.e.,
\begin{equation}
    \label{eq:sufficient_Td_constraint}
    \sum_{k=0}^{N}\tau_k+\sum_{k=0}^{N-1}2d_k\leq{\bar{T}_d}.
\end{equation}
Regarding the delivery deadline, we may require that
\begin{equation}
    \label{eq:sufficient_deadline}
    \sum_{k=0}^{N-1}\max\Big\{b_k\big(2d_k+p_k+\Bar{t}_k\big), \tilde{b}_k\big(2d_k+T_r\big)\Big\}\leq \Delta T,
\end{equation}
where 
$$\Bar{t}_k=\frac{e_f-e_s}{\min\big\{P_k,P_{\max}\big\}},$$
which is the maximum charging time needed at each station. 

When the above inequalities hold, we obtain a feasible base solution by setting $b_k\!=\!\Tilde{b}_k\!=\!1$ for all $k$. Note that conditions \eqref{eq:sufficient_Td_constraint} and \eqref{eq:sufficient_deadline} are likely to be restrictive, as they imply that constraints \eqref{eq:total_Td_constraint} and \eqref{eq:deadline} are never active. Still, the other conditions stated here may be used for the construction of a feasible base solution.

\subsection{Linearization of Bilinear Constraints}\label{app:f}
The bilinear terms $b_k\Delta e_k$, $k\!=\!0,\dots,N\!-\!1$ appear in the battery dynamics \eqref{eq:battery_energy}. They can be linearized by introducing additional variables $\Delta \hat{e}_k$ with constraints
\begin{align*}
    0\leq&\Delta \hat{e}_k\leq b_k\overline{\delta},\\
    0\leq& \Delta e_k-\Delta\hat{e}_k \leq \overline{\delta}(1-b_k),
\end{align*}
where $\overline{\delta}$ is a large positive constant used to approximate the unboundedness above, as in \eqref{eq:battery_energy}. The bilinear terms $b_kt_k$ appeared in the HoS regulation constraint \eqref{eq:Tr_constraint} and the deadline constraint \eqref{eq:deadline}, as well as the cost \eqref{Eq.14} can be linearized via an identical procedure. See, e.g., \cite[p.~176]{williams2013model} for further discussions.

For the product terms $\Tilde{b}_kb_k$ appeared in \eqref{Eq.2b}, we can introduce binary variables $\hat{b}_k\!\in\!\{0,1\}$ with constraints
$$\hat{b}_k\leq \Tilde{b}_k,\;\hat{b}_k\leq b_k,\;\hat{b}_k\geq \Tilde{b}_k+b_k-1,\quad k\!=\!0,\dots,N\!-\!1.$$
As for the constraints $\bar{b}_k\!=\!b_k\!\lor\!\tilde{b}_k$, they can be described in linear forms as
$$\bar{b}_k\geq \Tilde{b}_k,\;\bar{b}_k\geq b_k,\;\bar{b}_k\leq \Tilde{b}_k+b_k,\quad k\!=\!0,\dots,N\!-\!1.$$

With the above transformation, the mixed integer optimal charging problem becomes one with linear constraints and costs. However, even for the mixed integer linear program, the iterations needed for the exact solution may still grow exponentially with the problem scale. Moreover, if the linear approximation \eqref{eq:charging} for charging process $\Delta e_k$ is replaced by more accurate functions that are nonlinear in $t_k$, such a transformation does not lead to any simplifications. 
\bibliographystyle{IEEEtran}
\bibliography{LCSS_2023}
\end{document}